\makeatletter
\def\hgm@selfaliasguard#1#2#3{%
  \def\hgm@ctra{#2}\def\hgm@ctrb{#3}%
  \ifx\hgm@ctra\hgm@ctrb\else#1{#2}{#3}\fi}
\AddToHook{package/thm-autoref/after}{%
  \ifdefined\newcounteralias
    \let\hgm@newcounteralias@a\newcounteralias
    \def\newcounteralias{\hgm@selfaliasguard\hgm@newcounteralias@a}%
  \fi}
\AddToHook{package/aliascnt/after}{%
  \ifdefined\newcounteralias
    \let\hgm@newcounteralias@b\newcounteralias
    \def\newcounteralias{\hgm@selfaliasguard\hgm@newcounteralias@b}%
  \fi}
\makeatother

\documentclass[a4paper,USenglish,cleveref,autoref,thm-restate,numberwithinsect]{lipics-v2021}

\pdfoutput=1
\hideLIPIcs

\title{Online Stochastic Matchings:\texorpdfstring{\\}{ }Stability on Hypergraphs}
\titlerunning{Online stochastic matchings: stability on hypergraphs}

\author{Fabien Mathieu}{Swapcard, Paris, France \and Sorbonne Universit\'e, CNRS, LIP6, Paris, France \and \url{https://balouf.github.io/}}{fabien@swapcard.com}{https://orcid.org/0000-0003-1362-0359}{}

\authorrunning{F. Mathieu}
\Copyright{Fabien Mathieu}



\acknowledgements{The author thanks C\'eline Comte and Sushil Mahavir Varma for their outstanding
contributions to the polytope perspective~\cite{CMV25} from which this work grew.}

\nolinenumbers
\allowdisplaybreaks

\newcommand{\V}{V}
\newcommand{\Ed}{E}
\newcommand{\N}{\mathbb{N}}
\newcommand{\Z}{\mathbb{Z}}
\newcommand{\R}{\mathbb{R}}
\newcommand{\E}{\mathbb{E}}
\newcommand{\PP}{\mathbb{P}}
\newcommand{\cone}{\operatorname{cone}}
\newcommand{\interior}{\operatorname{int}}
\newcommand{\supp}{\operatorname{supp}}
\newcommand{\ml}{\textsf{ML}}
\newcommand{\vqml}{\textsf{VQML}}
\newcommand{\fcfs}{\textsf{FCFS}}
\newcommand{\emptystate}{\varnothing}
\newcommand{\candyaxisheight}{7.8cm}

\usepackage{tikz}
\usetikzlibrary{backgrounds,arrows.meta}
\usepackage{pgfplots}
\pgfplotsset{compat=1.16}
\usepackage{subcaption}

\begin{document}

\maketitle

\begin{abstract}
In stochastic dynamic matching, items of several classes arrive at rates $\lambda$ and leave in compatible groups. 
When groups may exceed two items, as in quantum switches or organ exchange, compatibility is described by the incidence matrix $A$ of a hypergraph, whose entries count how many items of each class a group consumes (single-class hyperedges can model discards). 
We call a system \emph{stabilizable} if some matching policy keeps large queues rare in the long run, and show that this holds if and only if the conservation equation $A\mu = \lambda$ admits a nonnegative solution whose support induces a surjective submatrix. Equivalently, $\lambda$ lies in the interior of the cone generated by the hyperedges.

The condition is necessary for all policies, including idling and clairvoyant ones. It is sufficient through one explicit policy: the rewardless counterpart of the Extended Greedy Primal--Dual policy of Nazari and Stolyar (2019) stabilizes every stabilizable system without knowledge of the rates, and its mean queue length is explicitly bounded.
Approaching a solution whose support is not surjective has a cost, under every policy: matching rates within $\epsilon$ of it force mean queues of order at least $1/\epsilon$.

On simple graphs, some greedy policies stabilize every stabilizable system.
On a small hypergraph with a single hyperedge of size three, every greedy policy is unstable on an open set of stabilizable rates. This shows that no greedy policy, nor any family of greedy policies, can approach maximal stability on hypergraphs: idling is necessary.
We also settle two cases explicitly: the complete $k$-uniform hypergraph (the $k$-partite quantum switch), which is stabilizable exactly when no link generates a fraction $1/k$ or more of the total rate, and the Fano plane.

\end{abstract}

\section{Introduction}\label{sec:intro}

In a stochastic dynamic matching system, items of finitely many \emph{classes} arrive at
random over time, wait in per-class queues, and leave in compatible groups. A fixed family
of hyperedges specifies which multisets of classes may be matched and cleared together.
Groups of three items or more arise in assemble-to-order systems and multi-way resource pooling, 
and in the two applications that motivate this paper.

In the continuous-time model of a quantum entanglement switch~\cite{NVGT20,VGNT21}, each user
is connected to the switch by a link that generates entangled qubit pairs as a Poisson
process, and the switch stores its half of each pair; a joint measurement on stored qubits
of $k$ distinct users consumes them and, when it succeeds, leaves those users sharing a
$k$-partite entangled state. Users are then classes and stored qubits are items: a switch
that serves every group of $k$ users is a matching system on the complete $k$-uniform
hypergraph. 

In living-donor organ exchange, incompatible patient--donor pairs swap donors
along cycles of two or more pairs and along chains started by an altruistic donor, for
kidneys~\cite{RSU04,AAGK17} and, more recently, livers~\cite{YSU24}: pair types are
classes and admissible cycles are hyperedges.
\Cref{fig:exchange} shows a toy pool where one pair type can only be served by a three-way cycle.

\begin{figure}[t]
\centering
\definecolor{ex1}{HTML}{E377C2}
\definecolor{ex2}{HTML}{FF7F0E}
\definecolor{ex3}{HTML}{9467BD}
\definecolor{ex4}{HTML}{D62728}
\definecolor{ex5}{HTML}{2CA02C}
\tikzset{
  pair/.style={draw=black!65, fill=black!6, rounded corners, thick, align=center,
    minimum width=1.2cm, minimum height=0.8cm, inner sep=2pt, font=\footnotesize},
  strand/.style={line width=1.2pt},
  tip/.style={-{Latex[length=2.4mm, width=2.8mm]}, black},
  hyp/.style={draw=#1, fill=#1!15, thick, minimum width=0.8cm, minimum height=0.35cm},
  inc/.style={line width=1pt},
}
\begin{subfigure}[b]{0.52\linewidth}
\centering
\begin{tikzpicture}
\node[pair] (n1) at (0,2) {AB/B\\ $\lambda = 4$};
\node[pair] (n2) at (-2.3,0) {A/B\\ $\lambda = 2$};
\node[pair] (n3) at (0,0) {B/A\\ $\lambda = 4$};
\node[pair] (n4) at (2.3,0) {A/AB\\ $\lambda = 1$};
\begin{scope}[on background layer]
  \foreach \c/\dx in {ex4/-7.1pt, ex5/-5pt, ex2/-2.9pt}
    \draw[strand, \c] ([xshift=\dx, yshift=1pt]n1.south) -- ([xshift=\dx, yshift=2.4mm]n3.north);
  \draw[strand, ex2] ([xshift=4pt, yshift=-1pt]n3.north) -- ([xshift=4pt, yshift=-2.4mm]n1.south);
  \foreach \c/\dy in {ex4/5.05pt, ex3/2.95pt}
    \draw[strand, \c] ([xshift=1pt, yshift=\dy]n3.west) -- ([xshift=2.4mm, yshift=\dy]n2.east);
  \draw[strand, ex3] ([xshift=-1pt, yshift=-4pt]n2.east) -- ([xshift=-2.4mm, yshift=-4pt]n3.west);
  \draw[strand, ex5] ([xshift=-1pt]n3.east) -- ([xshift=-2.4mm]n4.west);
  \draw[strand, ex4] ([yshift=-1pt]n2.north) |- ([xshift=-2.4mm]n1.west);
  \draw[strand, ex5] ([yshift=-1pt]n4.north) |- ([xshift=2.4mm]n1.east);
\end{scope}
\draw[tip] ([xshift=-5pt, yshift=2.4mm]n3.north) -- ([xshift=-5pt]n3.north);
\draw[tip] ([xshift=4pt, yshift=-2.4mm]n1.south) -- ([xshift=4pt]n1.south);
\draw[tip] ([xshift=2.4mm, yshift=4pt]n2.east) -- ([yshift=4pt]n2.east);
\draw[tip] ([xshift=-2.4mm, yshift=-4pt]n3.west) -- ([yshift=-4pt]n3.west);
\draw[tip] ([xshift=-2.4mm]n4.west) -- (n4.west);
\draw[tip] ([xshift=-2.4mm]n1.west) -- (n1.west);
\draw[tip] ([xshift=2.4mm]n1.east) -- (n1.east);
\draw[strand, ex1, -{Latex[length=2.4mm, width=2.8mm, color=black]}]
  ([xshift=-2.8mm]n1.north) .. controls +(100:9mm) and +(80:9mm) .. ([xshift=2.8mm]n1.north);
\end{tikzpicture}
\caption{Compatibility digraph. An arc from $i$ to~$j$ means that the donor of~$i$ can give
to the recipient of~$j$; each color shows one feasible exchange cycle.}
\label{fig:exchange-digraph}
\end{subfigure}%
\hfill
\begin{subfigure}[b]{0.46\linewidth}
\centering
\begin{tikzpicture}
\node[pair] (c2) at (4.35,2.2) {A/B\\ $\lambda = 2$};
\node[pair] (c3) at (2.9,2.2) {B/A\\ $\lambda = 4$};
\node[pair] (c1) at (1.45,2.2) {AB/B\\ $\lambda = 4$};
\node[pair] (c4) at (0,2.2) {A/AB\\ $\lambda = 1$};
\foreach \k/\x in {5/0, 1/1.0875, 2/2.175, 4/3.2625, 3/4.35}
  \node[hyp=ex\k] (h\k) at (\x,0) {};
\foreach \k/\members in {1/{1}, 2/{1,3}, 3/{2,3}, 4/{1,2,3}, 5/{1,3,4}}
  \foreach \m in \members
    \draw[inc, ex\k] (h\k.north) -- (c\m.south);
\end{tikzpicture}
\caption{Compatibility hypergraph. The hyperedges (bottom) are the cycles of~(\subref{fig:exchange-digraph}); each one is joined to the classes (top) it contains.}
\label{fig:exchange-hyper}
\end{subfigure}
\caption{A toy organ-exchange pool under blood-type compatibility. Each class is a type of
recipient--donor pair, labeled by the blood types of the recipient and of the donor and by
its arrival rate~$\lambda$; AB/B is a self-compatible pair, the other three are not.
The A/AB class belongs to a single cycle, of length three (it can only give to AB/B and receive from B/A).
With pairwise exchanges only, A/AB pairs never leave. With three-way cycles, $\mu = 1$ on every hyperedge solves
$A\mu = \lambda$ and $A$ has rank~$4$, so the pool is stabilizable (\cref{thm:main}).}
\label{fig:exchange}
\end{figure}

The first question about any such system, before rewards or delays, is \emph{stabilizability}: does some matching policy keep the queues from blowing up?

On \emph{simple graphs}, where every edge joins two distinct classes, the problem is solved: a connected graph is stabilizable for some arrival rates exactly when it is non-bipartite~\cite{MM16}, and the region is cut out by independent-set (Hall-type) inequalities~\cite{MM16,C22,BMM21}. 
The polytope perspective of~\cite{CMV25} recasts this linear-algebraically: writing $A$ for the class-edge incidence matrix and $\lambda$ for the arrival-rate vector, $(G,\lambda)$ is stabilizable if and only if $\lambda$ lies in the interior of the cone generated by the edges, equivalently the conservation equation $A\mu=\lambda$ admits a nonnegative solution whose support induces a surjective submatrix~\cite[Prop.~3]{CMV25}. 
Greedy policies, which match items as soon as they can, suffice there: match-the-longest~\cite{MM16} and first-come-first-matched~\cite{MBM21} are \emph{maximally stable}, that is, stable whenever some policy is, although not all greedy policies are~\cite{MP17}.

On hypergraphs, the question was open apart from some specific families. Rahme and Moyal focus on greedy policies and expect progress ``only on a case-by-case basis''~\cite{RM21}. The polytope paper outlines the extension and leaves its proof to the present one~\cite[\S7.2]{CMV25}.
No greedy policy can be maximally stable on hypergraphs: on some of them, every greedy policy is unstable at rates where an idling policy, i.e.\ a policy that may hold items back for a future group, is stable.

\paragraph*{Contributions.}
\begin{itemize}
	\item \emph{A characterization of stabilizability for hypergraphs} (\cref{thm:main}). The
	characterization recalled above for simple graphs holds on hypergraphs: $(G,\lambda)$ is
	stabilizable if and only if the conservation equation admits a nonnegative solution whose
	support induces a surjective submatrix. Stability may be read as the absence of explosion
	in time average, for any policy, or as positive recurrence of a Markov chain, as in the
	matching literature: both readings give the same stabilizable problems.
	Hyperedges may repeat a class; mono-edges (hyperedges of size one) model items that can leave alone (discarded or matched on their own).
	\begin{itemize}
		\item The condition is necessary for every policy, idling or not, by the pathwise argument
		of~\cite{MM16}; outside the region, large queues are even
		the rule in time average, under every policy (\cref{prop:necessity}).
		Extending~\cite[Prop.~10]{CMV25} beyond positive recurrent policies and simple graphs,
		every limit point of the mean matching rates of a stable policy is itself a solution of
		this kind (\cref{prop:rates}). Approaching
		a solution whose support is not surjective has a cost, under every policy: rates within
		$\epsilon$ of it force mean queues of order at least $1/\epsilon$ (\cref{cor:degenerate}); on simple graphs, this settles the first-moment lower bound
		conjectured in~\cite{CMV25}.
		\item It is sufficient through a single policy defined from $G$ alone, Virtual-Queue
		Match-the-Longest (\vqml{})~\cite{CMV25}, the rewardless counterpart of the Extended Greedy Primal--Dual
		policy of Nazari and Stolyar~\cite{NS19}: with their drift analysis, \vqml{} stabilizes
		every stabilizable instance, hence is maximally stable, with an explicit mean-queue bound
		(\cref{prop:vqml-stable}).
	\end{itemize}
	\item \emph{Differences with simple graphs} (\cref{sec:candy,sec:pagoda}).
	\begin{itemize}
		\item \emph{No greedy policy is maximally stable} (\cref{sec:candy}). We introduce the \emph{candy}, a
	seven-class hypergraph with a single hyperedge of size three where every greedy policy is
	unstable on an open set of stabilizable rates (\cref{prop:candy-open}), the
	non-anticipative ones on a much wider range (\cref{prop:candy-greedy});
	to our knowledge, it is the first stabilizable instance on which every greedy policy is
	unstable. No family of greedy policies can therefore approach maximal stability.
	\item \emph{No simple basis of the kernel} (\cref{prop:pagoda}). On simple graphs, the kernel has a basis of even cycles and
kayak paddles, with weights $\pm 1$ and $\pm 2$. This does not extend: in the \emph{pagodas}, a family of hypergraphs with at most three classes per hyperedge and no multiplicity, the kernel is a line spanned by an exchange whose weight ratio is $2^{n/3}$ (\cref{fig:pagoda}).
	\end{itemize}
	\item \emph{Stability regions for specific families} (\cref{sec:families}).
	\begin{itemize}
		\item The complete $k$-uniform hypergraph, that is, the $k$-partite quantum switch, is
		stabilizable if and only if no link generates a fraction $1/k$ or more of the total rate
		(\cref{prop:complete}). Under the switching rule of that device, this region was known for
		identical links~\cite{NVGT20}, for $k=2$~\cite{VGNT21} and for $k=3$~\cite[Thm.~1]{RM21}:
		idling cannot enlarge it; for $k \geq 4$ with unequal links the characterization is new.
		\item Similar computations settle the Fano plane (\cref{prop:fano}) and the symmetric
		designs. For the complete uniform hypergraphs and the Fano plane, they settle, for general
		policies, the analogue of conjectures that Rahme's thesis~\cite{Rahme21} states for greedy ones; the greedy versions remain open in general.
	\end{itemize}
\end{itemize}

\Cref{sec:related} reviews related work. \Cref{sec:model} sets the model and states the
main theorem, and \cref{sec:examples} applies it: greedy policies on the candy, kernel
bases, and the exact regions of classical families. \Cref{sec:requirements} gathers what
stability requires of every policy: necessity, and the matching rates of stable policies.
\Cref{sec:sufficiency} proves sufficiency through \vqml{}, completes the proof of
\cref{thm:main} and compares it with concurrent work; the positive recurrence of \vqml{} is
deferred to Appendix~\ref{app:pr}. \Cref{sec:conclusion} lists open questions.

\section{Related work}\label{sec:related}

We review matching on graphs and its extension to hypergraphs, then the policies and
performance results we rely on, and the two applications that motivate the paper.

\paragraph*{Matching on graphs.}
The general (non-bipartite) stochastic matching model was introduced by Mairesse and
Moyal~\cite{MM16}. Beyond the results recalled in \cref{sec:intro}, Comte~\cite{C22} and
Begeot et al.~\cite{BMM21} give its stability region, with product-form and
ubiquitous-measure descriptions. The model extends to graphs with self-loops, where two
items of the same class can be matched together, with necessary and sufficient stability
conditions and an explicit stationary law under first-come-first-matched~\cite{BMMR20}.
Jonckheere et al.~\cite{JMRS20} study generalized max-weight policies with noisy decisions
and impatient classes, and derive from a quadratic Lyapunov function explicit bounds on the
stationary mean of the largest queue. The bipartite case, non-stabilizable on its own but
central to the First-Come-First-Served (\fcfs{}) matching-rate theory, goes back to~\cite{CKW09,BGM13}, and relates to
FCFS parallel service systems~\cite{AKRW18}. The polytope/conservation viewpoint we build on
is that of~\cite{CMV25}. Two properties hold on simple graphs but fail on hypergraphs
(\cref{sec:examples}): some greedy policies, like match-the-longest~\cite{MM16} or
first-come-first-matched~\cite{MBM21}, are maximally stable; and the kernel of the incidence
matrix has a basis of even cycles and kayak paddles~\cite[\S4.1.3]{CMV25}.

\paragraph*{Multi-way matching and hypergraphs.}
Rahme and Moyal~\cite{RM21} initiated the stability study of matching on hypergraphs. They
give necessary conditions in terms of transversals, rank and anti-rank, and exhibit broad
families of non-stabilizable geometries (e.g.\ a hyperedge with two degree-one classes, or
$r$-uniform ``bipartite'' hypergraphs), together with the exact region for complete
$3$-uniform hypergraphs. Two features leave the general question open, and motivate the
present work: their policies are \emph{greedy} (matching is mandatory on arrival), so the
possibility that idling enlarges the stabilizable set is outside their framework; and no
single policy is shown or conjectured maximally stable, the authors expecting progress case
by case~\cite[\S7]{RM21}. Rahme's thesis~\cite{Rahme21} pursues this study, and conjectures
the region of the complete uniform hypergraphs and the stabilizability of the Fano plane
(\cref{sec:families}). Open problems of the field are surveyed in~\cite{MM22}.

\paragraph*{Policies and performance.}
Nazari and Stolyar~\cite{NS19} introduced the Extended Greedy Primal--Dual (EGPD) policy for reward maximization. Its rewardless counterpart, Virtual-Queue Match-the-Longest (\vqml{}), was introduced in~\cite{CMV25}; together with the drift analysis of EGPD, it gives our sufficiency half, for which we claim no novelty (\cref{sec:sufficiency} details what is imported). 
Kerimov, Ashlagi and Gurvich~\cite{KAG23,KAG24} characterize and achieve
constant regret under a general-position condition (a unique, non-degenerate optimum of
the static planning linear program), and show greedy policies hindsight-optimal in two-way
networks but not beyond; Gupta~\cite{G24} gives a greedy multi-way policy with bounded
regret, and Wei, Xu and Yu~\cite{WXY23} a constant-regret primal-dual policy for multi-way
matching. Gurvich and Ward~\cite{GW14} study the dynamic control of matching queues, and
for the N-shaped bipartite graph with linear holding costs the optimal control is of
threshold type on one edge, hence idles below the threshold~\cite{CDB22}. These policies
are either two-way, or $\lambda$-dependent, or tuned to a reward objective; none is a
single $\lambda$-oblivious maximally stable policy, which is what the characterization
requires.

\paragraph*{Organ exchange.}
In living-donor organ exchange, incompatible patient--donor pairs swap donors along cycles,
or along chains that start from an altruistic donor, or in some programs from a deceased
donor~\cite{MRLRR16}, and end with a patient on the waiting list~\cite{RSU04,AAGK17}. In
the present model, pair and donor types are classes, admissible cycles and chains are
hyperedges, and self-compatible pairs or exits to the waiting list are mono-edges. For
kidneys, adding three-way cycles to pairwise exchanges yields substantial gains, while
longer cycles add less~\cite{RSU07}; programs cap the cycle length differently, France for
instance allowing up to six pairs since 2021~\cite{CHT22}. Liver exchange, long limited to
pairwise swaps~\cite{ESU20}, has recently performed cycles of up to six pairs~\cite{YSU24},
and kidneys and livers can be traded within one pool~\cite{DS17}. Most dynamic models of
exchange pools aim at minimizing the waiting time rather than characterizing stability:
under random compatibility, greedy matching is near optimal and three-way cycles shorten
the mean wait~\cite{AAGK17}, while, with compatibility types, holding pairs back for a
larger future exchange can pay~\cite{Unv10}, a form of idling. \Cref{thm:main} is used
in~\cite{Mat27} on exchanges under blood-type compatibility, to certify that altruistic or
deceased donors and bounded discards, both mono-edges, restore the stability that a
symmetric distribution of pairs rules out.

\paragraph*{Quantum switches.}
With identical links and no decoherence, the entanglement switch has been analyzed
exactly: capacity and mean memory occupancy of the bipartite switch in discrete
time~\cite{VGNT20}, and, for $k$-partite states, a stability criterion (positive recurrence
if and only if there are more users than parties) with closed forms for the capacity and
the mean number of stored qubits~\cite{NVGT20}, with a finer analysis of the tripartite
case in~\cite{NVGT22}. Heterogeneous links, finite memories and
decoherence are treated for the bipartite switch in~\cite{VGNT21}, and switches that serve
both pairs and triples, under randomized policies that may hold qubits back for a future
triple, in~\cite{VNGT23}. A second family of models queues the entanglement requests rather
than the stored qubits, characterizes the capacity region and gives policies that achieve
it, such as MaxWeight scheduling~\cite{VT22,DRT22,PVTT23,PVGT24}. When stored qubits
decohere, they abandon: for a three-node switch serving requests for two of the three pairs,
each match consuming a request and two qubits, MaxWeight on the request queues, which is greedy, is throughput optimal among greedy policies~\cite{ZJM22}.
For a general switch topology with decoherence, multipartite requests included, Bhambay,
Vasantam and Walton~\cite{BVW25} characterize the capacity region, show that MaxWeight is then no longer
throughput optimal, and give an asymptotically throughput-optimal scheduler derived from an
average-reward Markov decision process.
\Cref{sec:families} gives the stabilizable region of the $k$-partite switch for links of
unequal rates.

\section{Model and main result}\label{sec:model}

\subsection{Model and policies}\label{sec:policies}

Let $\V = \{1, \dots, n\}$ be a finite set of \emph{classes} and
$\Ed = \{1, \dots, m\}$ a finite set of \emph{hyperedges}, described by an incidence
matrix $A \in \N^{n \times m}$: activating hyperedge $k$ consumes $A_{i,k}$ items of
class~$i$, for each $i \in \V$. We write $A_k \in \N^n$ for the $k$-th column, $e_i$ for
the $i$-th unit vector of~$\R^n$, and assume $A_k \neq 0$ for all $k$. Mono-edges ($A_k = e_i$, a class that departs on its own) and
repeated entries ($A_{i,k} \geq 2$, several items of one class cleared at once) are allowed.
We write $G = (\V, \Ed)$ for the resulting hypergraph.
Items of class $i$ arrive according to independent Poisson processes with rates
$\lambda_i > 0$, one item at a time; $\Lambda = \sum_{i\in \V} \lambda_i$ is the total rate and
$\bar\lambda = \lambda / \Lambda$ is the vector of normalized rates, so that arrival classes at successive epochs are
i.i.d.\ with law $\bar\lambda$. For $S \subseteq \V$, $\lambda(S) = \sum_{i \in S} \lambda_i$,
so that $\lambda(\V) = \Lambda$, and $\mathbf 1_S$ is the indicator vector of~$S$; $\mathbf 1$
is the all-ones vector, of $\R^n$ or $\R^m$ as the context requires. The pair $(G, \lambda)$ is called
a \emph{hypergraph matching problem}. Unmatched items wait in per-class queues
$X(t) \in \N^n$. A \emph{matching policy} decides, at arrival epochs, which hyperedges
to activate: each activation of $k$ removes the multiset $A_k$ from the queues, and is
allowed only if the required items are available. The decision may depend on anything (past, independent randomness, even future arrivals), and the policy may idle. The system \emph{starts empty}, $X(0) = 0$.
The tuple $(G, \lambda, \Phi)$, where $\Phi$ is a matching policy, is called a \emph{matching model}.

\label{rem:discard}A mono-edge $A_k=e_i$ is a \emph{controlled discard} of class~$i$: the policy may remove a
class-$i$ item on its own. Adding the mono-edges $\{e_i : i\in S\}$ for a set
$S\subseteq\V$ therefore models discarding (loss) on $S$, and \cref{thm:main} applies verbatim
with $A$ extended by those columns. Discarding only enlarges $\cone(A)$, so it can only help:
for $S=\V$ the mono-edges already span $\R^n_{\ge0}$, hence
$\lambda\in\R^n_{>0}\subseteq\interior\cone(A)$ for every $\lambda$ and the problem is
stabilizable unconditionally. This is the familiar fact that abandonment everywhere trivializes stability.
However, for a strict subset $S\subsetneq\V$ the criterion is genuine: the non-discardable classes must still be balanced by matches. Partial discarding is thus covered by the same characterization, with no separate theory, and connects to the loss-queue viewpoint of~\cite{C22}.

Three properties of policies recur. A policy is \emph{non-anticipative} if each of its
decisions is a function of the arrivals up to its epoch and of randomness independent of
the arrival sequence, and \emph{clairvoyant} if it may use future arrivals. It is
\emph{greedy} if, whenever an arriving item can be part of an activable hyperedge, some
hyperedge is activated at that arrival,\footnote{This is the non-idling
sense. It excludes commit-and-reserve policies such as the bounded-regret greedy
of~\cite{G24}, which assigns arrivals to optimal-basis configurations without matching them
on the spot (closer to \vqml{}'s virtual backlog than to non-idling matching) and uses
knowledge of $\lambda$; \cref{prop:candy-greedy,prop:candy-open} do not constrain those.}
and \emph{idling} otherwise. Markovian policies (\cref{sec:stability_notions}) are
non-anticipative.

Let $a(t) \in \{e_1, \dots, e_n\}$ be the arrival at epoch~$t$ (i.i.d.\ with
law~$\bar\lambda$), $s_k(t) \in \N$ the number of activations of hyperedge~$k$ at
epoch~$t$, $N(t) = \sum_{u<t} a(u)$ the past arrivals and $M_k(t) = \sum_{u<t} s_k(u)$ the past matches. Whatever the policy, the queue $X$ shows the difference between arrivals and departures:
\begin{equation}\label{eq:balance}
X(t) \;=\; X(0) + N(t) - A\,M(t) \;\geq\; 0 .
\end{equation}

\subsection{Stability notions}\label{sec:stability_notions}

\Cref{thm:main} below involves two related notions of stability: a weak one, which applies to every model, and a stronger one, which concerns Markovian policies.

\paragraph*{Cesàro stability.}

We say that a matching model is \emph{stable} if large queues are rare in the long run:
\begin{equation}\label{eq:stable}
	\lim_{K \to \infty}\ \limsup_{T \to \infty}\ \frac1T \sum_{t < T}
	\PP\bigl(\|X(t)\|_1 \geq K\bigr) \;=\; 0 ,
\end{equation}
that is, the long-run fraction of epochs at which the system holds $K$ items or more
vanishes as $K$ grows. The problem $(G, \lambda)$ is \emph{stabilizable} if some policy
makes it stable, and a policy is \emph{maximally stable} if it stabilizes every
stabilizable problem.

\paragraph*{Markov stability.}
We say that a policy is \emph{Markovian} if it has a \emph{countable Markov structure}, in the following sense.
The policy can use an auxiliary internal state $Y(t)$ that ranges over a fixed countable set. 
Both the activation decision and the update of $Y$ occur at arrival epochs only, and depend only on the queue state $X(t)$, the internal state $Y(t)$, the arriving class $a(t)$, and possibly fresh independent randomness.
Thus $S(t) := (X(t), Y(t))$ is a time-homogeneous Markov chain on a countable state space, and the continuous-time process is its constant-rate-$\Lambda$ uniformization. Randomization is permitted in the transition kernel, but a continuously distributed quantity may not be \emph{stored} in $Y$ (that would leave the countable setting).
Its scope is discussed at the end of this subsection.
Matching the convention of~\cite[Appendix~A]{CMV25}, we assume that the internal state starts at one distinguished initial value $Y(0)=y_0$, and call $\emptystate := (0, y_0)$, the initial state of $S$,
the \emph{empty state}.
The corresponding model is \emph{positive recurrent} if the embedded chain is positive recurrent on the communicating class (i.e., the strongly connected component) of $\emptystate$, or equivalently if the continuous-time chain is.

Positive recurrence implies stability in the sense of~\labelcref{eq:stable}: starting from the empty state, the Ces\`aro averages of the laws of the embedded chain converge pointwise to the stationary law~$\pi$ of its class~\cite{B99}, hence in total variation by Scheff\'e's lemma, so that $\frac1T \sum_{t<T} \PP(\|X(t)\|_1 \geq K) \to \pi(\|x\|_1 \geq K)$, a quantity that vanishes as $K \to \infty$. In the sequel, \emph{stable} always refers to~\labelcref{eq:stable} and \emph{positive recurrent} to the Markov notion just defined.

\paragraph*{Scope of the considered policies.}
	The class of Markovian policies defined above \emph{strictly extends} the policy model of \cite[Appendix~A]{CMV25}: it allows several activations per epoch, matchings among already-queued items (not involving the arriving item), and 
	drops the assumptions that $\emptystate$ is the only state with empty queues, and of irreducibility 
	(the latter must then be \emph{proven} for a given policy, not assumed). Neither the stability notion~\labelcref{eq:stable} nor the necessity part of \Cref{thm:main} (\cref{prop:necessity}) needs the class: the former is defined, and the latter holds, for every policy. The Markovian class matters for condition~(i$'$), for the convergence of the matching rates in \cref{prop:rates}, and for Appendix~\ref{app:pr}, where \vqml{} is shown to belong to it (\cref{prop:vqml-pr}).
	No bound on the number of activations per epoch is needed: each activation removes at least one item, so the activations never outnumber the arrivals ($M_k(T) \leq T$ for every hyperedge~$k$ and epoch~$T$), and this is all that the convergence of the matching rates in \cref{prop:rates} uses.
	Countability of the internal state keeps the definition of positive recurrence elementary: positive recurrence makes sense on a countable state space, whereas a continuous internal state would call for Harris recurrence. 
	
	Nothing is lost by these restrictions, since necessity holds for \emph{every} policy and \vqml{}, the policy that supports sufficiency, is Markovian.

\subsection{Characterization of stability}\label{sec:stability}

The \emph{conservation equation} $A\mu = \lambda$, where $\mu \in \R_{\geq 0}^m$ is a vector of matching rates, balances departures against arrivals. We write
$\cone(A) = A\,\R^m_{\geq 0}$ for the cone spanned by the columns of $A$, and call $A$
\emph{surjective} when $A \R^m = \R^n$ (trivial left kernel).

\begin{theorem}\label{thm:main}
For a hypergraph matching problem $(G, \lambda)$ with $\lambda \in \R_{>0}^n$, the
following are equivalent:
\begin{enumerate}
\item[(i)] $(G, \lambda)$ is stabilizable: some policy makes it stable in the sense
of~\labelcref{eq:stable};
\item[(i$'$)] \emph{(Markov form)} Some Markovian policy makes it positive recurrent;
\item[(ii)] For every $y \in \R^n \setminus \{0\}$ such that $\langle y, A_k \rangle \leq 0$ for
all $k \in \Ed$, we have $\langle y, \lambda \rangle < 0$;
\item[(iii)] The conservation equation admits a solution $\mu$ with $\mu_k > 0$ for
all $k$, and $A$ is surjective; equivalently, $\lambda \in \interior \cone(A)$;
\item[(iii$'$)] \emph{(support form)} The conservation equation admits a solution
$\mu \in \R_{\geq 0}^m$ such that the restriction of $A$ to the columns in
$\supp(\mu)$ is surjective.
\end{enumerate}
A single policy, \vqml{} (see \cref{sec:vqml-policy}, with the tie-breaking rule of Appendix~\ref{app:pr}), defined from $G$ alone with no
knowledge of $\lambda$ and no parameter, makes $(G,\lambda)$ positive recurrent, hence
stable, whenever these conditions hold: \vqml{} is maximally stable.
\end{theorem}

\emph{Proof structure.} Conditions (ii), (iii) and (iii$'$) are equivalent by convex
duality (\cref{sec:geometry}). The necessity half, (i)\,$\Rightarrow$\,(ii), is proved in
\cref{sec:necessity}, and the sufficiency half, (iii)\,$\Rightarrow$\,(i), in
\cref{sec:sufficiency}. Appendix~\ref{app:pr} gives (iii)\,$\Rightarrow$\,(i$'$). Finally,
(i$'$)\,$\Rightarrow$\,(i) is trivial: positive recurrence implies stability
(\cref{sec:stability_notions}).

\vqml{} is defined in \cref{sec:vqml-policy}; until then, the following description
suffices. The policy plans matchings ahead of the items. It keeps a signed ledger with one
counter per class: a positive counter counts the waiting items in excess of what the
planned matchings need, a negative one the items that they still miss. At
each arrival it plans two copies of the hyperedge whose classes have the largest total
counter, with multiplicity, if that total is positive, and nothing otherwise; the planned
matchings are executed, oldest first, as soon as the waiting items can complete them. The policy
idles by design: it keeps items waiting until the matchings that it has planned can be completed.

Numerically, condition (iii) is decided by one linear program and a rank computation: with
$t^\ast = \max\{t : A\mu = \lambda,\ \mu_k \geq t \text{ for all } k\}$, the problem
$(G,\lambda)$ is stabilizable if and only if $\operatorname{rank} A = n$ and $t^\ast > 0$.
An optimal $\mu$ also supplies the slack $\delta = t^\ast/\Lambda$ of the explicit
bound of \cref{prop:vqml-stable}. 
Condition (iii$'$) is a convenient alternate form for checking stabilizability: it asks only for positivity on a surjective subset of hyperedges rather than on all hyperedges. It is useful to prove stabilizability on (hyper)graphs where the (hyper)edges are very large or only defined implicitly.  In the simple graph case it reads: \emph{some nonnegative solution of $A\mu=\lambda$ has a support all of whose connected components are non-bipartite.}

Stability in~(i) is a weak reading of ``the queues do not blow up'': it asks that large queues be rare in time average. It makes sense for any policy, Markovian or not. 
Positive recurrence, in~(i$'$), is the notion of the matching literature~\cite{MM16,CMV25}. It implies~(i), and the theorem says that both readings define the same stabilizable problems. Positive recurrence buys a proper stationary law: convergent matching rates (\cref{prop:rates}), explicit stationary formulas through the normalizing constant $\pi(\emptystate) = 1/\E_\emptystate[\tau_\emptystate]$, as in the
product forms of~\cite{MBM21,BMMR20,C22}, and the mean-queue bound in stationarity rather
than in time average (\cref{prop:vqml-pr}). Its proof needs one more ingredient, the
accessibility of the empty state, which depends on how \vqml{} breaks ties
(Appendix~\ref{app:pr}).

Even this weak notion keeps the mean queue sublinear, a property used in
\cref{sec:candy,sec:rates}.

\begin{lemma}[Stable models have sublinear mean queues]\label{lem:sublinear}
If $(G, \lambda, \Phi)$ is stable in the sense of~\labelcref{eq:stable}, then
$\E\|X(T)\|_1 = o(T)$ as $T \to \infty$.
\end{lemma}

\begin{proof}
One item arrives per epoch, activations only remove items and the system starts empty,
so $\|X(t)\|_1 \geq \|X(T)\|_1 - (T - t)$ for every $t \leq T$, and $\|X(T)\|_1 \leq T$.
Let $K \geq 0$ be an integer. If $\|X(T)\|_1 \geq K$, then $\|X(t)\|_1 \geq K$ at each of
the $\|X(T)\|_1 - K$ epochs $t$ such that $T - \|X(T)\|_1 + K \leq t < T$. Hence,
pathwise, $\|X(T)\|_1 \leq K + \#\{t < T : \|X(t)\|_1 \geq K\}$, and
$\E\|X(T)\|_1/T \leq K/T + \frac1T \sum_{t<T} \PP(\|X(t)\|_1 \geq K)$. Letting
$T \to \infty$, then $K \to \infty$, the right-hand side vanishes
by~\labelcref{eq:stable}.
\end{proof}

\subsection{Equivalent forms of the condition}\label{sec:geometry}

Conditions (ii), (iii) and (iii$'$) describe one geometric fact, that $\lambda$ lies in the
interior of the cone spanned by the hyperedges: (ii) from the dual side, by separating
directions, (iii) and (iii$'$) from the primal side, by positive combinations of
hyperedges.

\begin{lemma}[Convex duality]\label{lem:geometry}
For $A \in \R^{n \times m}$ and $\lambda \in \R^n$, the following are equivalent:
\begin{enumerate}
\item[(a)] $\lambda \in \interior \cone(A)$;
\item[(b)] $A$ is surjective and $A\mu = \lambda$ for some $\mu \in \R_{>0}^m$;
\item[(b$'$)] $A\mu = \lambda$ for some $\mu \in \R_{\geq 0}^m$ such that the restriction of
$A$ to the columns in $\supp(\mu)$ is surjective;
\item[(c)] $\langle y, \lambda\rangle < 0$ for every $y \in \R^n \setminus \{0\}$ such that
$\langle y, A_k\rangle \leq 0$ for all~$k$.
\end{enumerate}
\end{lemma}

\begin{proof}
\emph{(b) $\Rightarrow$ (b$'$)} is immediate.

\emph{(b$'$) $\Rightarrow$ (c).} Let $y \neq 0$ with $\langle y, A_k\rangle \leq 0$ for
all~$k$. Then
$\langle y, \lambda\rangle = \sum_{k \in \supp(\mu)} \mu_k \langle y, A_k\rangle$ is a sum
of nonpositive terms. If it vanished, $y$ would be orthogonal to every column of
$\supp(\mu)$, hence to $\R^n$, which these columns span: impossible.

\emph{(c) $\Rightarrow$ (a).} Suppose that $\lambda \notin \interior\cone(A)$. The set
$\cone(A)$ being convex, the separation theorem yields a direction $y \neq 0$ with
$\langle y, x \rangle \leq \langle y, \lambda \rangle$ for all $x \in \cone(A)$. A linear
form bounded above on a cone is nonpositive on it, so $\langle y, A_k \rangle \leq 0$ for
every~$k$; and $x = 0$ gives $\langle y, \lambda \rangle \geq 0$: (c) fails.

\emph{(a) $\Rightarrow$ (b).} A cone with nonempty interior spans $\R^n$, so $A$ is
surjective. Let $w = A\mathbf{1}$. For $\varepsilon > 0$ small enough,
$\lambda - \varepsilon w \in \cone(A)$, that is, $\lambda - \varepsilon w = A \nu$ with
$\nu \geq 0$, and $\lambda = A(\nu + \varepsilon \mathbf{1})$ with
$\nu + \varepsilon\mathbf{1} > 0$.
\end{proof}

\paragraph*{Consistency with simple graphs.}
When every column has two distinct unit entries (simple graphs), (ii) recovers the classical
Hall-type independent-set condition for graphs~\cite{MM16}: for an independent set $\mathcal{I}$ with neighborhood
$\Gamma(\mathcal I)$, the vector $y = \mathbf{1}_{\mathcal I} - \mathbf{1}_{\Gamma(\mathcal I)}$
satisfies $\langle y, A_k\rangle \leq 0$ for every edge $k$ (edges inside $\mathcal I$ do not
exist; edges from $\mathcal I$ go to $\Gamma(\mathcal I)$ and contribute $0$; all other
edges contribute $0$, $-1$ or $-2$), and $\langle y, \lambda\rangle < 0$ reads
$\lambda(\mathcal I) < \lambda(\Gamma(\mathcal I))$. Bipartiteness of a component yields
$y = \pm(\mathbf{1}_{\V^+} - \mathbf{1}_{\V^-})$ with $\langle y, A_k\rangle = 0$ for every edge~$k$: each side is an
independent set whose neighborhood is the other side, the inequalities
$\lambda(\V^+) < \lambda(\V^-)$ and $\lambda(\V^-) < \lambda(\V^+)$ cannot both hold, and (ii)
fails for one of the two sides. Condition (ii) thus subsumes both non-bipartiteness and
the independent-set inequalities. The converse, that these classical conditions (for connected
graphs) imply (ii), is the simple graph characterization of~\cite{MM16}.

\section{Examples and applications}\label{sec:examples}

Before proving \cref{thm:main}, we show two ways in which hypergraphs differ from simple
graphs beyond the characterization itself: greedy policies can fail to stabilize a
stabilizable instance (\cref{sec:candy}), and the kernel of the incidence matrix need not
have a simple basis (\cref{sec:pagoda}). We then derive from the theorem the exact stability
regions of classical families (\cref{sec:families}).

\subsection{Greedy policies: the candy}\label{sec:candy}

We illustrate the first difference on a small but instructive hypergraph, which we call the
\emph{candy}.
The candy makes concrete the greedy--idling gap exposed in \cref{sec:intro}: on it, \emph{every} greedy (non-idling) policy fails to stabilize some stabilizable instances, so idling is necessary.
This is specific to hypergraphs: on simple graphs, match-the-longest~\cite{MM16} and first-come-first-matched~\cite{MBM21} are greedy and maximally stable, so idling is never needed there; on hypergraphs, no greedy policy is maximally stable (\cref{prop:candy-open}).
We first describe the candy and its stability region, then a one-parameter family of rates
on which greedy policies fail, and finally compare with the closest precedents.

\paragraph*{The candy.}

The \emph{candy}, shown in \cref{fig:candy}, has $n=7$ classes and $m=7$ hyperedges: two simple triangles on $\{1,2,3\}$ and $\{5,6,7\}$, joined by a single central hyperedge $\{3,4,5\}$, whose private class $4$ has degree one (it belongs to no other edge).

With edge order
$\{1,2\},\{1,3\},\{2,3\},\{5,6\},\{5,7\},\{6,7\},\{3,4,5\}$, the incidence matrix is
\[
A=\begin{pmatrix}
1&1&0&0&0&0&0\\
1&0&1&0&0&0&0\\
0&1&1&0&0&0&1\\
0&0&0&0&0&0&1\\
0&0&0&1&1&0&1\\
0&0&0&1&0&1&0\\
0&0&0&0&1&1&0
\end{pmatrix}.
\]

\begin{figure}[t]
\centering
\begin{tikzpicture}[
  cls/.style={circle,draw,fill=white,inner sep=1pt,minimum size=6mm,font=\small},
  thick]
\begin{scope}[on background layer]
  \fill[black!25] (0,0) ellipse (2.4cm and 0.8cm);
  \draw[black!75,dashed] (0,0) ellipse (2.4cm and 0.8cm);
\end{scope}
\node[cls] (n1) at (-3.3, 0.9) {$1$};
\node[cls] (n2) at (-3.3,-0.9) {$2$};
\node[cls] (n3) at (-1.6, 0)   {$3$};
\node[cls] (n4) at ( 0,   0)   {$4$};
\node[cls] (n5) at ( 1.6, 0)   {$5$};
\node[cls] (n6) at ( 3.3, 0.9) {$6$};
\node[cls] (n7) at ( 3.3,-0.9) {$7$};
\draw (n1)--(n2) (n1)--(n3) (n2)--(n3);
\draw (n5)--(n6) (n5)--(n7) (n6)--(n7);
\node[font=\footnotesize] at (0,-1.2) {hyperedge $\{3,4,5\}$};
\end{tikzpicture}
\caption{The candy: two simple triangles, $\{1,2,3\}$ and $\{5,6,7\}$, joined by the hyperedge
$\{3,4,5\}$ (shaded), whose private class~$4$ has degree one. Simple edges are drawn as
segments, the hyperedge as a region enclosing its three classes.}
\label{fig:candy}
\end{figure}

It is square and invertible, so $A\mu=\lambda$ has a unique solution $\mu(\lambda)$ for
every $\lambda$, and by \cref{thm:main} (surjectivity holds since $\det A\neq 0$)
$(G,\lambda)$ is stabilizable iff that solution is strictly positive. Solving the two
triangles against the central edge gives the following closed form.

\begin{proposition}[Stability region of the candy]\label{prop:candy-region}
The candy is stabilizable if and only if
\[
|\lambda_1-\lambda_2| \;<\; \lambda_3-\lambda_4 \;<\; \lambda_1+\lambda_2
\qquad\text{and}\qquad
|\lambda_7-\lambda_6| \;<\; \lambda_5-\lambda_4 \;<\; \lambda_6+\lambda_7 .
\]
\end{proposition}

\begin{proof}
Class $4$ has the hyperedge as its only edge, so $\mu_{\{3,4,5\}}=\lambda_4$. On the
first triangle the remaining balance equations are
$\mu_{\{1,2\}}+\mu_{\{1,3\}}=\lambda_1$, $\mu_{\{1,2\}}+\mu_{\{2,3\}}=\lambda_2$ and
$\mu_{\{1,3\}}+\mu_{\{2,3\}}=\lambda_3-\lambda_4$, whose unique solution is
$\mu_{\{1,2\}}=\tfrac12(\lambda_1+\lambda_2-(\lambda_3-\lambda_4))$,
$\mu_{\{1,3\}}=\tfrac12(\lambda_1-\lambda_2+(\lambda_3-\lambda_4))$,
$\mu_{\{2,3\}}=\tfrac12(\lambda_2-\lambda_1+(\lambda_3-\lambda_4))$. These are all
positive iff $|\lambda_1-\lambda_2|<\lambda_3-\lambda_4<\lambda_1+\lambda_2$; the second
triangle is symmetric. By \cref{thm:main} positivity of $\mu(\lambda)$ is equivalent to
stabilizability.
\end{proof}

\Cref{prop:candy-region} shows that the bridge class $3$ (and similarly class $5$) must both feed the central hyperedge (at rate $\lambda_4$, leaving
$\lambda_3-\lambda_4$) and absorb its triangle's imbalance $|\lambda_1-\lambda_2|$; when
the leftover cannot cover the imbalance, $(G,\lambda)$ leaves the cone. For instance
$\lambda=(1.5,1,1.1,1,1.1,1,1.5)$ is \emph{not} stabilizable: here
$\lambda_3-\lambda_4=0.1<0.5=|\lambda_1-\lambda_2|$, the forced solution has
$\mu_{\{2,3\}}=\mu_{\{5,6\}}=-0.2$, and \cref{thm:main}(ii) is violated by
$y=\mathbf 1_{\{1,4\}}-\mathbf 1_{\{2,3\}}$, for which $\langle y, A_k\rangle\le 0$ for every edge
yet $\langle y,\lambda\rangle=0.4>0$.

\paragraph*{A stabilizable family.}
Fix $\alpha\in(0,1)$ and take
$\lambda_\alpha=(1,1,3\alpha,\alpha,3\alpha,1,1)$. Both conditions of
\cref{prop:candy-region} hold ($0<2\alpha<2$), so the candy is stabilizable for
\emph{every} $\alpha\in(0,1)$ (but not for $\alpha = 1$, where $\lambda_\alpha$ reaches
the boundary of the cone; $\alpha = 0$ leaves the model, three classes having zero
rate); the unique matching-rate vector puts $1-\alpha$ on the two
outer edges $\{1,2\},\{6,7\}$ and $\alpha$ on the other five edges (including the
hyperedge). By \cref{thm:main}, \vqml{} (defined in \cref{sec:sufficiency}) stabilizes the
whole family.

\paragraph*{Instability of greedy policies.}
Match-the-longest, which activates, among the edges that the arrival completes, one whose
other classes hold the most items in total, is Markovian and greedy. We write $N_4(T)$ for
the number of class-$4$ arrivals before epoch~$T$.

\begin{proposition}[Instability of non-anticipative greedy policies]\label{prop:candy-greedy}
Let $\alpha_{\mathrm{g}}:=\tfrac{1+\sqrt{17}}{12}\approx 0.427$. On the candy with rates
$\lambda_\alpha$, $\alpha<\alpha_{\mathrm{g}}$, no non-anticipative greedy policy is stable, Markovian
or not: a positive fraction of the class-$4$ items is left waiting,
\[
\liminf_{T\to\infty}\frac{\E[X_4(T)]}{\E[N_4(T)]} \;\ge\;
\frac{2+3\alpha-18\alpha^2}{2+3\alpha+9\alpha^2} \;>\; 0 .
\]
\end{proposition}

\begin{proof}
Under a greedy policy, anticipating or not, no edge is ever activable from the queued
items alone; in other words, (a) in each triangle at most one of the three queues is
positive, and (b) $X_3,X_4,X_5$ are not all positive. This holds at the empty start. If it
holds just before an arrival, every edge that becomes activable contains the arriving
item: either there is none, and the item joins its queue; or the policy, being greedy,
activates one, which consumes the arrival and leaves no queue above its value before the
arrival. In both cases no edge is activable after the epoch, and at most one edge is
activated per epoch.

Weigh the first triangle by
\[
\psi \;=\; (2+3\alpha)\,X_3-3\alpha\,(X_1+X_2) .
\]
By (a), an arrival finds this triangle in exactly one of three situations; suppose first
that it does not fire the hyperedge. If $X_3>0$, then $X_1=X_2=0$: a class-$1$ or
class-$2$ arrival completes a single edge and removes a class-$3$ item, so that $\psi$
loses $2+3\alpha$, and a class-$3$ arrival waits, so that $\psi$ gains $2+3\alpha$. If
$X_1=X_2=X_3=0$, an arrival of class $1$, $2$ or $3$ waits: $\psi$ loses $3\alpha$ for
class $1$ or $2$ and gains $2+3\alpha$ for class~$3$. If $X_1+X_2>0$, then one of the two
queues is empty and $X_3=0$: an arrival of the class in stock waits, and $\psi$ loses
$3\alpha$; an arrival of the other of the two classes, or of class~$3$, removes one item
of the stock, and $\psi$ gains $3\alpha$. The other classes do not move $\psi$.

The hyperedge fires only when an arrival completes it: a class-$3$ arrival that finds
$X_4,X_5>0$, hence $X_3=0$ by (b); a class-$5$ arrival that finds $X_3,X_4>0$; or a
class-$4$ arrival that finds $X_3,X_5>0$. In the first case $\psi$ stays put instead of
gaining $2+3\alpha$ or $3\alpha$; in the other two it loses $2+3\alpha$, as a class-$3$
item is removed. Compared with the moves above, a firing thus lowers $\psi$ by at least
$3\alpha$ if class~$3$ completes it, and by $2+3\alpha$ otherwise. The weight
$\psi'=(2+3\alpha)\,X_5-3\alpha\,(X_6+X_7)$ of the second triangle moves in the same
way, with classes $3$ and $5$ exchanged, so a firing lowers $\psi+\psi'$ by at least
$2+6\alpha$.

Now fix a non-anticipative greedy policy, so that the class of an arrival is independent
of the queues that it finds. Weighted by the rates $1$, $1$ and $3\alpha$ of classes $1$,
$2$ and $3$, the moves of $\psi$ without firing add up to $9\alpha^2-4$ if the arrival
finds $X_3>0$, and to $9\alpha^2$ in the other two situations. With
$p_t=\PP(X_3(t)>0)+\PP(X_5(t)>0)$ and $f_t$ the probability that the hyperedge fires at
epoch~$t$, this gives
\[
\Lambda\,\E\bigl[(\psi+\psi')(t+1)-(\psi+\psi')(t)\bigr]
\;\le\; 18\alpha^2-4\,p_t-(2+6\alpha)\,\Lambda f_t .
\]
By (b), the three situations in which the hyperedge fires exclude one another; each has
$X_3>0$ or $X_5>0$, and the completing class has rate at most $3\alpha$. Hence
$\Lambda f_t\le 3\alpha\,p_t$, and the right-hand side is at most
$18\alpha^2-c\,\Lambda f_t$, with $c=2+6\alpha+\tfrac{4}{3\alpha}$. Let $F(T)$ be the
number of firings before epoch~$T$, of mean $\sum_{t<T}f_t$. Summing over $t<T$ from the
empty start, with $\psi\ge-3\alpha\,U$ for the stock $U=X_1+X_2$ and
$\psi'\ge-3\alpha\,U'$ for $U'=X_6+X_7$,
\[
c\,\E[F(T)] \;\le\; 18\alpha^2\,\frac{T}{\Lambda}+3\alpha\,\E\bigl[U(T)+U'(T)\bigr] .
\]
The stocks are small. When $U$ leaves $0$, $U^2$ gains $1$; from $U=k\ge 1$, it gains
$2k+1$ at an arrival of the class in stock, loses $2k-1$ at an arrival of the other
class, which is as likely, and can only lose at a class-$3$ arrival. The mean increment
of $U^2$ per epoch is thus at most $2/\Lambda$, so that
$\E[U(T)]^2\le\E[U(T)^2]\le 2T/\Lambda$, and likewise for $U'$. As
$18\alpha^2/c=27\alpha^3/(2+3\alpha+9\alpha^2)$ and $6\alpha\le c$,
\[
\E[F(T)] \;\le\; \frac{27\alpha^2}{2+3\alpha+9\alpha^2}\cdot\frac{\alpha T}{\Lambda}
+\sqrt{2T/\Lambda} .
\]
Class~$4$ leaves the system only through the hyperedge, so
$\E[X_4(T)]=\E[N_4(T)]-\E[F(T)]$ with $\E[N_4(T)]=\alpha T/\Lambda$, and the announced
bound follows. It is positive when $18\alpha^2<2+3\alpha$, that is, when
$\alpha<\alpha_{\mathrm{g}}$. Finally, a stable model has $\E\|X(T)\|_1=o(T)$ by
\cref{lem:sublinear}, so the policy is not stable.
\end{proof}

\begin{proposition}[Open set of stabilizable rates incompatible with greedy matching]\label{prop:candy-open}
On the candy with rates $\lambda \in \R^7_{>0}$, let
\[
\rho(\lambda) = \frac{\lambda_3\lambda_5\,(\lambda_1+\lambda_2+\lambda_4+\lambda_6+\lambda_7)}
{(\lambda_1+\lambda_2)(\lambda_6+\lambda_7)} .
\]
Under every greedy policy, even a clairvoyant one, aware of future arrivals,
\[
\liminf_{T\to\infty} X_4(T)/N_4(T) \geq 1 - \rho(\lambda)/\lambda_4 \text{ almost surely.}
\]
In particular, no greedy policy is stable when $\rho(\lambda) < \lambda_4$. The stabilizable
rates that satisfy this inequality form a nonempty open set, which contains
$\lambda_\alpha$ for every
$\alpha < \alpha_{\mathrm{c}} := \tfrac{2}{3}(\sqrt{10}-3) \approx 0.108$.
\end{proposition}

\begin{proof}
The invariants (a) and (b) of the previous proof hold for every greedy policy, anticipating
or not. Let the walk $W_3$ start at~$0$, go up by one at class-$3$ arrivals and down by
one, when positive, at arrivals of classes $1$ and $2$. As $X_3$ goes up at class-$3$
arrivals only and, by~(a), goes down at every arrival of class $1$ or $2$ that finds it
positive, $X_3 \leq W_3$ at every epoch; likewise $X_5 \leq W_5$, with classes $5$, $6$
and~$7$. The hyperedge thus fires only at epochs where a class-$3$ arrival finds
$W_5 > 0$, a class-$5$ arrival finds $W_3 > 0$, or a class-$4$ arrival finds both
positive. The number $D(T)$ of such epochs before epoch~$T$ depends on the arrivals alone.

Almost surely, the long-run fraction of epochs with $W_3 > 0$ is at most
$r_3 = \lambda_3/(\lambda_1+\lambda_2)$: when this ratio is below one, the walk is positive
recurrent and $r_3$ is its stationary probability of being positive; otherwise there is
nothing to prove. Likewise for $W_5$, with $r_5 = \lambda_5/(\lambda_6+\lambda_7)$.
Moreover, the fraction of epochs where both walks are positive is at most $r_3 r_5$. An
arrival moves at most one of the two walks, so the pair $(W_3, W_5)$ is the
uniformization, at rate~$\Lambda$, of two independent M/M/1 queues: when both ratios are
below one, its stationary law is the product of their two geometric laws. Otherwise, if
$r_3 \geq 1$ say, the fraction is at most that of the epochs with $W_5 > 0$, hence at most
$r_5 \leq r_3 r_5$. As $W_3$ and $W_5$ at epoch~$t$ are functions of the arrivals
before~$t$, of which the arrival at epoch~$t$ is independent, the strong law for bounded
martingale differences gives, almost surely,
\[
\limsup_{T\to\infty} \frac{D(T)}{T} \;\leq\;
\frac{\lambda_3 r_5 + \lambda_5 r_3 + \lambda_4 r_3 r_5}{\Lambda}
\;=\; \frac{\rho(\lambda)}{\Lambda} .
\]
Class~$4$ leaves the system only through the hyperedge, so
$X_4(T) \geq N_4(T) - D(T)$, and
$N_4(T)/T \to \lambda_4/\Lambda$ almost surely, whence the first claim. If
$\rho(\lambda) < \lambda_4$, the central queue tends to infinity almost surely, so that
$\PP(\|X(t)\|_1 \geq K) \to 1$ for every $K$, and~\labelcref{eq:stable} fails. Finally,
$\rho$ is continuous and the stability region of \cref{prop:candy-region} is open, so the
stabilizable rates such that $\rho(\lambda) < \lambda_4$ form an open set; it contains
$\lambda_\alpha$, which is stabilizable, for every $\alpha < \alpha_{\mathrm{c}}$, as
$\rho(\lambda_\alpha) = \tfrac94\alpha^2(4+\alpha)$ is less than $\alpha$ if and only if
$9\alpha^2 + 36\alpha - 4 < 0$.
\end{proof}

\Cref{prop:candy-open} rules out maximal stability for greedy policies: a maximally stable
policy stabilizes every stabilizable problem, and no greedy policy does so on the candy, a
single hypergraph with a single hyperedge of size three. No family of greedy policies,
however parametrized, can approach maximal stability either: the rates that some greedy
policy stabilizes miss a nonempty open set of stabilizable rates. On that set, idling is
not a refinement traded against performance, but the only way to be stable.

The two thresholds differ in quality. We do not expect $\alpha_{\mathrm{c}}\approx 0.108$ to
be tight, even for clairvoyant policies: the domination ignores the stocks
of the outer classes, which the proof of \cref{prop:candy-greedy} exploits. The threshold
$\alpha_{\mathrm{g}}\approx 0.427$ leaves little slack: match-the-longest is greedy and
non-anticipative, and simulations put its own threshold near
$\alpha_{\mathrm{ML}}\approx 0.45$. \Cref{fig:candy-sim}
plots the mean central queue $\overline{X_4}$ and the global delay under match-the-longest and under \vqml{}
across the family ($10^7$ arrivals per point, see~\cite{candynb} for details).

\definecolor{darkgray176}{RGB}{176,176,176}
\definecolor{darkorange25512714}{RGB}{255,127,14}
\definecolor{gray}{RGB}{128,128,128}
\definecolor{lightgray204}{RGB}{204,204,204}
\definecolor{steelblue31119180}{RGB}{31,119,180}

\begin{figure}[!htp]
\centering
\begin{subfigure}{\linewidth}
\centering
\begin{tikzpicture}	
	\begin{axis}[width=12cm, height=\candyaxisheight,
		legend cell align={left},
		legend style={fill opacity=0.8, draw opacity=1, text opacity=1, draw=lightgray204},
		log basis y={10},
		tick align=outside,
		tick pos=left,
		x grid style={darkgray176},
		xlabel={$\alpha$},
		xmajorgrids,
		xmin=0, xmax=1,
		xminorgrids,
		xtick style={color=black},
		y grid style={darkgray176},
		ylabel={$\overline{X_4}$ (mean size of central queue)},
		ymajorgrids,
		ymin=0.001, ymax=1000,
		ymode=log,
		ytick style={color=black}
		]
		\addplot [semithick, steelblue31119180]
		table [header=false, row sep=\\] {%
			0.46 64.08306\\0.47 29.7398972\\0.48 19.0804183\\0.49 14.2234257\\0.5 11.0414084\\0.51 9.1910888\\0.52 7.6832416\\0.53 6.6846177\\0.54 5.9376752\\0.55 5.2624664\\0.56 4.7453828\\0.57 4.2903827\\0.58 3.8958317\\0.59 3.551055\\0.6 3.2582841\\0.61 2.9926003\\0.62 2.7771704\\0.63 2.5656076\\0.64 2.3786882\\0.65 2.2216691\\0.66 2.0677091\\0.67 1.9152655\\0.68 1.7856184\\0.69 1.6601754\\0.7 1.5545498\\0.71 1.452527\\0.72 1.3563517\\0.73 1.2618299\\0.74 1.1714391\\0.75 1.0971712\\0.76 1.0238625\\0.77 0.9493186\\0.78 0.8804672\\0.79 0.8168422\\0.8 0.7564674\\0.81 0.6960784\\0.82 0.6396078\\0.83 0.5861354\\0.84 0.5345674\\0.85 0.4855473\\0.86 0.4399087\\0.87 0.3947679\\0.88 0.3551458\\0.89 0.3160284\\0.9 0.2783648\\0.91 0.244219\\0.92 0.2101807\\0.93 0.1801052\\0.94 0.1493978\\0.95 0.122585\\0.96 0.0939943\\0.97 0.0683535\\0.98 0.044741\\0.99 0.022537\\};
		\addlegendentry{\ml{} (greedy)}
		\addplot [semithick, darkorange25512714]
		table [header=false, row sep=\\] {%
			0.01 58.5291416\\0.02 30.0631547\\0.03 19.8769737\\0.04 15.2575022\\0.05 12.4161453\\0.06 10.5361278\\0.07 9.135511\\0.08 8.0783569\\0.09 7.2493809\\0.1 6.5735255\\0.11 6.0063501\\0.12 5.5491406\\0.13 5.1641332\\0.14 4.822137\\0.15 4.5161837\\0.16 4.2566181\\0.17 4.0229949\\0.18 3.8079322\\0.19 3.6168354\\0.2 3.4400009\\0.21 3.2795084\\0.22 3.1349379\\0.23 2.999115\\0.24 2.8806288\\0.25 2.7648786\\0.26 2.6563545\\0.27 2.555308\\0.28 2.4586647\\0.29 2.384749\\0.3 2.2955471\\0.31 2.2130826\\0.32 2.1335212\\0.33 2.0625276\\0.34 1.9917764\\0.35 1.924316\\0.36 1.8613907\\0.37 1.7990656\\0.38 1.740435\\0.39 1.6831978\\0.4 1.628878\\0.41 1.5755468\\0.42 1.524754\\0.43 1.4734065\\0.44 1.4260978\\0.45 1.3798935\\0.46 1.3345496\\0.47 1.2885799\\0.48 1.2481541\\0.49 1.2044803\\0.5 1.1625599\\0.51 1.123734\\0.52 1.0854269\\0.53 1.0500558\\0.54 1.0121719\\0.55 0.9756782\\0.56 0.9410177\\0.57 0.9064665\\0.58 0.872048\\0.59 0.839644\\0.6 0.8056607\\0.61 0.7739063\\0.62 0.7405958\\0.63 0.7079\\0.64 0.6772079\\0.65 0.6477234\\0.66 0.6184334\\0.67 0.5887726\\0.68 0.5599868\\0.69 0.5319193\\0.7 0.5037574\\0.71 0.4764818\\0.72 0.4516052\\0.73 0.4263037\\0.74 0.4002429\\0.75 0.375898\\0.76 0.3504033\\0.77 0.3258896\\0.78 0.3027577\\0.79 0.2800571\\0.8 0.2580462\\0.81 0.2376875\\0.82 0.2154755\\0.83 0.1954538\\0.84 0.1754852\\0.85 0.1557432\\0.86 0.136754\\0.87 0.1206557\\0.88 0.1042636\\0.89 0.0886918\\0.9 0.074564\\0.91 0.060864\\0.92 0.0487402\\0.93 0.0379048\\0.94 0.0280219\\0.95 0.0197654\\0.96 0.0125088\\0.97 0.006917\\0.98 0.0031325\\0.99 0.0008397\\};
		\addlegendentry{\vqml{}}
		\addplot [gray, dashdotted, line width=1.5pt]
		table [header=false, row sep=\\] {%
			0.108185106778920 0.0004002854320785\\0.108185106778920 4796.95771624986\\};
		\addlegendentry{$\alpha_{\mathrm{c}}$}
		\addplot [gray, dotted, line width=1.5pt]
		table [header=false, row sep=\\] {%
			0.426925468801472 0.0004002854320785\\0.426925468801472 4796.95771624986\\};
		\addlegendentry{$\alpha_{\mathrm{g}}$}
		\addplot [black, dashed, line width=1.5pt]
		table [header=false, row sep=\\] {%
			0.452 0.0004002854320785\\0.452 4796.95771624986\\};
		\addlegendentry{$\alpha_{\mathrm{ML}}$ (empirical)}
	\end{axis}
\end{tikzpicture}
\caption{Mean central queue $\overline{X_4}$.  For \ml{}, the observed central queue grows sharply as $\alpha$ decreases toward $\alpha_{\mathrm{ML}}$, whereas the simulated \vqml{} stays stable on the whole grid. Note that stabilizability fails in the limit $\alpha \to 0$ (three rates vanish and the problem degenerates into the two bipartite pairs $\{1,2\}$ and $\{6,7\}$), which explains why $\overline{X_4}$ grows for \vqml{} when $\alpha \rightarrow 0$.\label{fig:candy-q4}}
\end{subfigure}
\begin{subfigure}{\linewidth}
	\centering
\begin{tikzpicture}	
	\begin{axis}[width=12cm, height=\candyaxisheight,
		legend cell align={left},
		legend style={fill opacity=0.8, draw opacity=1, text opacity=1, draw=lightgray204},
		log basis y={10},
		tick align=outside,
		tick pos=left,
		x grid style={darkgray176},
		xlabel={$\alpha$},
		xmajorgrids,
		xmin=0, xmax=1,
		xminorgrids,
		xtick style={color=black},
		y grid style={darkgray176},
		ylabel={Delay (mean waiting time)},
		ymajorgrids,
		ymin=0.1, ymax=100,
		ymode=log,
		ytick style={color=black}
		]
		\addplot [semithick, steelblue31119180]
		table [header=false, row sep=\\] {%
			0.46 9.25929886426593\\0.47 4.45870013717421\\0.48 2.96788104619565\\0.49 2.28691039030956\\0.5 1.84262141333333\\0.51 1.5842235667107\\0.52 1.37578388743455\\0.53 1.23848687418936\\0.54 1.13656610539846\\0.55 1.04671626751592\\0.56 0.979649419191919\\0.57 0.922388410513141\\0.58 0.874882493796526\\0.59 0.835341156211562\\0.6 0.803828365853659\\0.61 0.777999661426844\\0.62 0.75998845323741\\0.63 0.74424242568371\\0.64 0.733340731132075\\0.65 0.72735114619883\\0.66 0.722980754060325\\0.67 0.721872566168009\\0.68 0.724896643835616\\0.69 0.730535968289921\\0.7 0.740922730337079\\0.71 0.753452296544036\\0.72 0.769320564159292\\0.73 0.789310823271131\\0.74 0.812172538126362\\0.75 0.839680886486486\\0.76 0.87249678111588\\0.77 0.910134760383387\\0.78 0.951207684989429\\0.79 0.999436946484785\\0.8 1.0514326875\\0.81 1.11173107549121\\0.82 1.18335070841889\\0.83 1.26710173292559\\0.84 1.36348053643725\\0.85 1.47353615075377\\0.86 1.60392170658683\\0.87 1.75251681863231\\0.88 1.92571084645669\\0.89 2.1352706744868\\0.9 2.3792389223301\\0.91 2.67413117647059\\0.92 3.05345865900383\\0.93 3.51464610846813\\0.94 4.11960989603025\\0.95 4.98097846948357\\0.96 6.26711641791045\\0.97 8.33138277108434\\0.98 12.1062694935543\\0.99 22.3255831747484\\};
		\addlegendentry{\ml{} (greedy)}
		\addplot [semithick, darkorange25512714]
		table [header=false, row sep=\\] {%
			0.01 36.1319094594595\\0.02 18.2707950724638\\0.03 11.8403064845606\\0.04 8.93026612149533\\0.05 7.16805303448276\\0.06 5.99005558823529\\0.07 5.12755024498886\\0.08 4.47818146929824\\0.09 3.96277799136069\\0.1 3.54929953191489\\0.11 3.20262815513627\\0.12 2.92374357438017\\0.13 2.69057230142566\\0.14 2.48650134538153\\0.15 2.3027303960396\\0.16 2.14766197265625\\0.17 2.00956337186898\\0.18 1.88470133079848\\0.19 1.77470242026266\\0.2 1.67505487037037\\0.21 1.5844526691042\\0.22 1.50423397111913\\0.23 1.42927832442068\\0.24 1.36366720070423\\0.25 1.30169311304348\\0.26 1.24430431271478\\0.27 1.19140971137521\\0.28 1.1424355704698\\0.29 1.10363424543947\\0.3 1.06057680327869\\0.31 1.02137873581848\\0.32 0.984544038461538\\0.33 0.951533312202853\\0.34 0.920341818181818\\0.35 0.891447891472868\\0.36 0.865171088957055\\0.37 0.839877116843703\\0.38 0.817031471471471\\0.39 0.795633982169391\\0.4 0.776238411764706\\0.41 0.757962983988355\\0.42 0.741461484149856\\0.43 0.725648159771755\\0.44 0.711857584745763\\0.45 0.699595286713287\\0.46 0.68877055401662\\0.47 0.678206159122085\\0.48 0.670157432065217\\0.49 0.662807940780619\\0.5 0.65582468\\0.51 0.651737252311757\\0.52 0.647841335078534\\0.53 0.646156044098573\\0.54 0.644093727506427\\0.55 0.644168777070064\\0.56 0.645999204545455\\0.57 0.648154042553191\\0.58 0.65233347394541\\0.59 0.65804504305043\\0.6 0.664805536585366\\0.61 0.6735123095526\\0.62 0.68301257793765\\0.63 0.695298002378121\\0.64 0.709086591981132\\0.65 0.724113005847953\\0.66 0.74034897911833\\0.67 0.760902681242808\\0.68 0.781633744292238\\0.69 0.805708595696489\\0.7 0.833136213483146\\0.71 0.862440880713489\\0.72 0.893505022123894\\0.73 0.930526410537871\\0.74 0.971247461873638\\0.75 1.0162620972973\\0.76 1.06690285407725\\0.77 1.12262249201278\\0.78 1.18182173361522\\0.79 1.24997951731375\\0.8 1.32307929166667\\0.81 1.40649195449845\\0.82 1.50330560574949\\0.83 1.61465070336391\\0.84 1.74302640688259\\0.85 1.88676775879397\\0.86 2.05350167664671\\0.87 2.24766896927651\\0.88 2.47024045275591\\0.89 2.74253196480938\\0.9 3.05339723300971\\0.91 3.43048407907425\\0.92 3.91422416666667\\0.93 4.51869861084681\\0.94 5.28763875236295\\0.95 6.40364661971831\\0.96 8.04042372201493\\0.97 10.6517267099166\\0.98 15.4644677900553\\0.99 28.7090556267155\\};
		\addlegendentry{\vqml{}}
		\addplot [gray, dashdotted, line width=1.5pt]
		table [header=false, row sep=\\] {%
			0.108185106778920 0.0004002854320785\\0.108185106778920 4796.95771624986\\};
		\addlegendentry{$\alpha_{\mathrm{c}}$}
		\addplot [gray, dotted, line width=1.5pt]
		table [header=false, row sep=\\] {%
			0.426925468801472 0.0004002854320785\\0.426925468801472 4796.95771624986\\};
		\addlegendentry{$\alpha_{\mathrm{g}}$}
		\addplot [black, dashed, line width=1.5pt]
		table [header=false, row sep=\\] {%
			0.452 0.0004002854320785\\0.452 4796.95771624986\\};
		\addlegendentry{$\alpha_{\mathrm{ML}}$ (empirical)}
	\end{axis}
	
\end{tikzpicture}
\caption{Average delay, obtained from average queue sizes by Little's law. Compared with \Cref{fig:candy-q4}, note the high delay for $\alpha\rightarrow 1$: $(G, \lambda_\alpha)$ is not stabilizable at $\alpha = 1$. The affected classes are the bridge classes 3 and 5, not the central class 4.\label{fig:candy-delay}}
\end{subfigure}
\caption{Candy family $\lambda_\alpha=(1,1,3\alpha,\alpha,3\alpha,1,1)$ under match-the-longest (greedy) and \vqml{} policies.
$10^7$ simulated arrivals per point. Match-the-longest is shown only where the run is
empirically stable. Instability is proven below $\alpha_{\mathrm{c}}\approx 0.108$ for every greedy policy, clairvoyant ones included (\cref{prop:candy-open}), and below $\alpha_{\mathrm{g}}\approx 0.427$ for the non-anticipative ones (\cref{prop:candy-greedy}); it is observed below $\alpha_{\mathrm{ML}} \approx 0.45$ for match-the-longest. The simulated \vqml{} is the implementation of the Stochastic Matching package~\cite{SM22}, a variant of \cref{eq:vqml}: it records each arrival in the virtual queue first, then plans one copy (not two) of the best-scoring hyperedge if its score is positive, and completes the oldest planned matching whose classes all have a waiting item, one per epoch. Produced in the companion notebook~\cite{candynb}.
	}
\label{fig:candy-sim}
\end{figure}

\paragraph*{Precedents.}
Instances of this phenomenon are known in other metrics. In a quantum switch whose
link-level entanglements decohere, a myopic policy such as MaxWeight serves requests that
need a single link as soon as their entanglement is available, and thereby starves a request
that needs three links, whereas holding entanglement back for it stabilizes higher request
rates~\cite[\S3.4]{BVW25}; decoherence acts there as abandonment. In a reward
setting, every greedy policy fails to achieve bounded
regret on a seven-type multi-way network with rates $\lambda, 2\lambda, \dots,
64\lambda$~\cite[Example~3.2]{KAG24} (see also~\cite{G24}): the agents needed for a
valuable four-way match are consumed by cheaper two-way matches before it can form. That
network has four matches for seven types, so some of its queues grow under every policy,
and the failure is one of value, not of stability; it is in fact driven by that growth,
since the over-supplied types $3$, $5$ and $7$ are always present and consume every
arriving agent of types $2$, $4$ and $6$. With discards on these three types, the incidence
matrix becomes square and invertible and the solution of the conservation equation is
positive, hence the problem is stabilizable, and in simulation greedy matching then
performs the four-way match at its full rate. The candy is square and invertible too, so
every stable policy has the same matching rates and rewards play no role. Its greedy
failure is structural: the triangles drain the bridge classes $3$ and $5$, which, under a
non-anticipative greedy policy, are present with probability $O(\alpha^2)$ only, so the
hyperedge fires at rate $O(\alpha^3)$ against the rate $\alpha$ it must absorb
(\cref{prop:candy-greedy}). The closest precedent we know of is a deterministic processing network
of Dai and Lin~\cite[\S6.2]{DL05}, in which two small activities starve an activity that
draws on all five buffers, including one that no other activity serves: under every
maximum-pressure policy, that buffer grows without bound from a suitable initial state,
whereas a policy that fires the large activity when all five buffers are nonempty, and
idles otherwise, is stable. Stability there only means that queues grow sublinearly, and,
read as a matching model, that network is not stabilizable: it has three edges for five
classes. To our knowledge the candy is the first stabilizable instance on which every greedy policy
is unstable; it involves no abandonment, which, on every class, would trivialize stability
(\cref{rem:discard}). It is small and fully controlled: seven classes with a
mirror symmetry, a single hyperedge of size three, a closed-form stability region, and an
instability threshold provable for \emph{every} greedy policy, which reaches
$\alpha \approx 0.427$ for the non-anticipative ones; in simulation,
match-the-longest fails near $\alpha = 0.45$, where all rates lie within a factor
three of each other. Unlike the non-stabilizable geometries of~\cite{RM21}
(\cref{sec:families}), on which no policy, greedy or not, can help, the candy is
stabilizable. It separates the combinatorial question answered by \cref{thm:main} from the
algorithmic behavior of greedy matching.

\subsection{Kernel bases: the pagodas}\label{sec:pagoda}

On simple graphs, $\ker A$, which describes the freedom of the matching rates under the constraint $A\mu=\lambda$, comes with a simple basis construction: from any bijective subgraph (for instance, a spanning tree augmented by an edge that closes an odd cycle, which exists whenever the graph is connected and surjective), the construction of~\cite[\S4.1.3]{CMV25} returns a basis of even cycles and kayak paddles, with weights $\pm 1$ and $\pm 2$. The choice of the subgraph only affects how sparse the basis is.

On hypergraphs, a nonzero integer vector of $\ker A$ still reads as a zero-sum exchange, $\sum_{k \in \Ed^+} a_k A_k = \sum_{k \in \Ed^-} b_k A_k$ with $\Ed^+ \cap \Ed^- = \emptyset$ and positive integers $a_k, b_k$: two ways of grouping the same items. In \cref{fig:exchange}, the three-way cycle through AB/B, A/B and B/A groups the same items as the self-compatible pair AB/B and the exchange between A/B and B/A.

Gaussian elimination from any $n$ independent columns still yields a basis of exchanges on hypergraphs, but no choice of columns, nor any other method, can keep the exchange weights bounded: the pagodas (\cref{fig:pagoda,prop:pagoda}) have a one-dimensional kernel, so every basis reduces to a multiple of a single exchange, whose weight ratio is $2^{n/3}$, although their hyperedges have at most three classes, without multiplicity.

\begin{figure}[!thb]
	\centering
	\definecolor{dP}{HTML}{1F77B4}
	\definecolor{dN}{HTML}{D62728}
	\def\pagw{0.8}\def\pagha{1.1}\def\paghb{0.9}
	\def\pagaa{-41.63}\def\pagab{221.63}\def\pagac{318.37}\def\pagin{0.541}
	\tikzset{
		dcls/.style={circle, draw, fill=white, inner sep=1pt, minimum size=6mm, font=\small},
		dpos/.style={thick, dP},
		dneg/.style={thick, dN},
		dhyp/.style={draw=dP, densely dashed, fill=dP!12, thick},
		dw/.style={font=\footnotesize, inner sep=1.5pt},
	}
	\newcommand{\pagoda}[2]{%
		\begin{scope}[shift={(#2,0)}]
			\node[font=\small, anchor=north] at (0, -0.6) {$P_{#1}$};
			\foreach \j in {1,...,#1} {
				\coordinate (l\j) at (-\pagw, {(\j-1)*(\pagha+\paghb)});
				\coordinate (r\j) at (\pagw, {(\j-1)*(\pagha+\paghb)});
				\coordinate (a\j) at (0, {(\j-1)*(\pagha+\paghb)+\pagha});
			}
			\begin{scope}[on background layer]
				\ifnum#1>1
				\foreach \j [evaluate=\j as \jn using int(\j+1)] in {1,...,\numexpr#1-1\relax} {
					\path[dhyp] ([shift={(\pagaa:0.45)}]a\j) -- ([shift={(\pagaa:0.45)}]r\jn)
						arc[start angle=\pagaa, end angle=90, radius=0.45] -- ([shift={(90:0.45)}]l\jn)
						arc[start angle=90, end angle=\pagab, radius=0.45] -- ([shift={(\pagab:0.45)}]a\j)
						arc[start angle=\pagab, end angle=\pagac, radius=0.45] -- cycle;
					\pgfmathtruncatemacro{\wf}{2^\j}
					\node[dw, text=dP] at (0, {.15+(\j-1)*(\pagha+\paghb)+\pagha+\pagin}) {$+\wf$};
				}
				\fi
				\path[dhyp] (a#1) circle[radius=0.45];
			\end{scope}
			\draw[dpos] (l1) -- (r1);
			\node[dw, anchor=north, text=dP] at (0, -0.08) {$+1$};
			\foreach \j in {1,...,#1} {
				\pgfmathtruncatemacro{\wu}{2^(\j-1)}
				\draw[dneg] (l\j) -- (a\j);
				\draw[dneg] (r\j) -- (a\j);
				\node[dw, anchor=east, text=dN] at ({-0.4*\pagw-0.1}, {(\j-1)*(\pagha+\paghb)+0.6*\pagha}) {$-\wu$};
				\node[dw, anchor=west, text=dN] at ({0.4*\pagw+0.1}, {(\j-1)*(\pagha+\paghb)+0.6*\pagha}) {$-\wu$};
			}
			\pgfmathtruncatemacro{\wm}{2^#1}
			\node[dw, anchor=south, text=dP] at ([yshift=0.45cm]a#1) {$+\wm$};
			\foreach \j in {1,...,#1} {
				\pgfmathtruncatemacro{\cl}{3*\j-2}
				\pgfmathtruncatemacro{\cr}{3*\j-1}
				\pgfmathtruncatemacro{\ca}{3*\j}
				\node[dcls] at (l\j) {$\cl$};
				\node[dcls] at (r\j) {$\cr$};
				\node[dcls] at (a\j) {$\ca$};
			}
	\end{scope}}
	\begin{tikzpicture}
		\pagoda{1}{0}
		\pagoda{2}{3.0}
		\pagoda{3}{6.0}
	\end{tikzpicture}
	\caption{The pagodas $P_1$, $P_2$ and~$P_3$ (\cref{prop:pagoda}). The pagoda $P_h$ stacks
		$h$ blocks of three classes, $n = 3h$ in all; its incidence matrix has a one-dimensional
		kernel, spanned by an exchange whose weight ratio is $2^h = 2^{n/3}$. Segments are
		hyperedges of size two, and shaded regions with a dashed border are hyperedges of size one
		or three; labels are signed weights, blue on one side of the exchange and red on the
		other. In the first block, raising the weight of $\{1,2\}$ by one implies lowering those of
		$\{1,3\}$ and $\{2,3\}$. This frees two items of class~$3$, which the mono-edge $\{3\}$
		absorbs with weight~$2$ in $P_1$. In $P_2$, the hyperedge $\{3,4,5\}$ absorbs them instead, which also consumes two items of classes~$4$ and~$5$, lowering $\{4,6\}$ and $\{5,6\}$ by two to restore balance and freeing four items of class~$6$. Classes $4$, $5$ and~$6$ thus form a second block, whose base is the hyperedge. The construction repeats indefinitely, each new block doubling the weights.
		}
	\label{fig:pagoda}
\end{figure}

\begin{proposition}[Pagodas]\label{prop:pagoda}
For $h \geq 1$, let $P_h$ be the hypergraph on the $n = 3h$ classes $1, \dots, 3h$ with the
hyperedges $\{1,2\}$, $\{3j-2, 3j\}$ and $\{3j-1, 3j\}$ for $1 \leq j \leq h$,
$\{3j, 3j+1, 3j+2\}$ for $1 \leq j < h$, and the mono-edge $\{3h\}$. Its incidence matrix
is surjective, and its kernel is the line spanned by the exchange, each hyperedge standing
for its column of~$A$,
\[
\begin{aligned}
&\{1,2\} + \sum_{j=1}^{h-1} 2^j\,\{3j, 3j+1, 3j+2\} + 2^h\,\{3h\}\\
&\qquad= \sum_{j=1}^{h} 2^{j-1} \bigl(\{3j-2, 3j\} + \{3j-1, 3j\}\bigr).
\end{aligned}
\]
Every basis of $\ker A$ is therefore a multiple of this exchange, whose weight ratio is
$2^h = 2^{n/3}$.
\end{proposition}

\begin{proof}
The relation holds class by class. Class $3j-2$ lies in $\{3j-2, 3j\}$, of weight $2^{j-1}$,
and in the hyperedge below it, $\{1,2\}$ or $\{3j-3, 3j-2, 3j-1\}$, of the same weight;
likewise for class $3j-1$. Class $3j$ lies in $\{3j-2, 3j\}$ and $\{3j-1, 3j\}$, of total
weight $2^j$, and in the hyperedge above it, $\{3j, 3j+1, 3j+2\}$ or $\{3h\}$, of
weight~$2^j$. The $3h$ hyperedges other than $\{1,2\}$ are independent: in a vanishing
combination, classes $1$ and~$2$ force the coefficients of $\{1,3\}$ and $\{2,3\}$ to
vanish, class~$3$ then forces that of the next hyperedge, and so on up the pagoda. Hence
$A$ has rank $n = 3h$, and its kernel, of dimension $(3h+1) - 3h = 1$, is spanned by the
exchange.
\end{proof}

\subsection{Exact stability regions of classical families}\label{sec:families}

We now put \cref{thm:main} to work. It first says which geometries are hopeless, and then
turns the stability regions of two families studied in~\cite{RM21} and in Rahme's
thesis~\cite{Rahme21} into short computations. In this subsection, $k$ is not the index of a hyperedge but
the uniformity, the common size of the hyperedges; a hyperedge is written~$e$, as a set
of classes, and its column~$A_e$.

\paragraph*{Non-stabilizable geometries.}

The non-stabilizable hypergraphs cataloged by Rahme and Moyal~\cite{RM21} share a mechanism
that the candy of \cref{sec:candy} already displays: a class that lies in a single hyperedge
is drained only jointly with the other classes of that edge, which couples their rates. The
extreme case is a lone hyperedge $\{1,2,3\}$: here $A$ has rank $1$, so $A$ is not
surjective and \cref{thm:main}(iii$'$) fails for every $\lambda$ (three classes cannot be
balanced by one control). Adding an edge $\{3,4\}$ (so that the hyperedge $\{1,2,3\}$ has
two degree-one classes, $1$ and $2$) leaves $A$ of rank $2<4$: the left-kernel vector
$y = \mathbf 1_{\{1\}} - \mathbf 1_{\{2\}}$ certifies non-stabilizability for every~$\lambda$, matching~\cite{RM21}. In each case the
obstruction is exactly the failure of the support-surjectivity of \cref{thm:main}(iii$'$);
the negative examples of~\cite{RM21} we have checked fail this condition or leave the cone,
and condition~(iii$'$) recovers their exactly-solved region for the complete $3$-uniform
case~\cite[Thm.~1]{RM21}.

\paragraph*{Complete uniform hypergraphs.}

\begin{proposition}[Complete uniform hypergraphs]\label{prop:complete}
Let $G$ be the complete $k$-uniform hypergraph on $n > k$ classes: every $k$-subset
of~$\V$ is a hyperedge. Then $(G,\lambda)$ is stabilizable if and only if
$k\,\lambda_i < \Lambda$ for every class~$i$.
\end{proposition}

Equivalently, $\lambda(\V \setminus \{i\}) > (k-1)\lambda_i$: the rest of the system must
arrive fast enough to supply the $k-1$ partners that each class-$i$ item needs in order to
leave.

\begin{proof}
We use condition~(ii) of \cref{thm:main}.

\emph{Only if.} For a class $i$ let $y = k\,\mathbf 1_{\{i\}} - \mathbf 1$. For a hyperedge $e$,
$\langle y, A_e\rangle = k\,\mathbf 1\{i \in e\} - k \leq 0$, and
$\langle y, \lambda\rangle = k\lambda_i - \Lambda$: if $k\lambda_i \geq \Lambda$, then $y$
violates~(ii).

\emph{If.} Let $y \neq 0$ be such that $\langle y, A_e\rangle \leq 0$ for every
hyperedge~$e$, and number the classes so that $y_1 \geq \dots \geq y_n$; for the hyperedge
$\{1, \dots, k\}$, the constraint reads $y_1 + \dots + y_k \leq 0$. Let
$w_i = k\lambda_i/\Lambda$, so that $0 < w_i < 1$ and $\sum_i w_i = k$, and let
\[
s \;=\; \sum_{i \leq k} (1 - w_i) \;=\; \sum_{i > k} w_i \;>\; 0 .
\]
As $y_i \leq y_{k+1}$ for $i > k$ and $y_i \geq y_k$ for $i \leq k$,
\[
\frac{k}{\Lambda}\,\langle y, \lambda\rangle - \sum_{i \leq k} y_i
\;=\; \sum_{i > k} w_i\, y_i - \sum_{i \leq k} (1 - w_i)\, y_i
\;\leq\; s\,y_{k+1} - s\,y_k \;\leq\; 0 ,
\]
so that $\langle y, \lambda\rangle \leq \frac{\Lambda}{k}\,(y_1 + \dots + y_k) \leq 0$.
Suppose that $\langle y, \lambda\rangle = 0$. Then all these inequalities are equalities:
the weights $w_i$ and $1 - w_i$ being positive, $y_i = y_{k+1}$ for $i > k$, $y_i = y_k$
for $i \leq k$, and $y_{k+1} = y_k$, so that $y$ is constant; and $y_1 + \dots + y_k = 0$,
so that $y = 0$, a contradiction. Hence $\langle y, \lambda\rangle < 0$, which is~(ii).
\end{proof}

\begin{remark}[What the boundary looks like]\label{rem:complete-boundary}
At $k\lambda_i = \Lambda$ the criterion fails structurally, not marginally. The certificate
$y = k\,\mathbf 1_{\{i\}} - \mathbf 1$ then has $\langle y, \lambda\rangle = 0$ while
$\langle y, A_e\rangle = -k\,\mathbf 1\{i \notin e\}$, so any $\mu \geq 0$ with
$A\mu = \lambda$ satisfies $0 = \langle y, A\mu\rangle = -k\sum_{e \not\ni i}\mu_e$ and is
therefore supported on the star of $i$. On that support $y$ lies in the left kernel, so
\cref{thm:main}(iii$'$) fails as well: every match has to go through class $i$, and the other
classes are left with no way of absorbing a fluctuation in which class-$i$ items arrive
slightly faster than $\lambda_i$. The same identity reads off what the slack measures:
$k\sum_{e \not\ni i}\mu_e = \Lambda - k\lambda_i$ holds for \emph{every} conservation
solution, so $(\Lambda - k\lambda_i)/k$ is the rate, common to all of them, at which the
system matches without touching class~$i$.
\end{remark}

The region has a probabilistic reading. A nonnegative solution of the conservation equation
has $\|\mu\|_1 = \Lambda/k$; divided by this mass, it is a probability distribution over
the hyperedges, that is, a fixed-size sampling design, whose first-order inclusion
probabilities are $k\lambda_i/\Lambda$, and the condition of \cref{prop:complete} is exactly
the requirement that these inclusion probabilities be below one. Systematic sampling, which
is the wrap-around rule of McNaughton~\cite{McN59}, constructs such a design with at most
$n$ hyperedges in its support: lay the classes end to end on $[0,\Lambda)$, class~$i$
occupying an interval of length~$\lambda_i$, cut the line into $k$ slots of length
$\Lambda/k$ and stack them; each height then selects one class per slot. For $n = k+1$ the
design is unique and explicit: indexing each hyperedge by the single class it omits makes
$A = \mathbf 1\mathbf 1^\top - I$, which is invertible, and
\[
\mu_{\V \setminus \{i\}} \;=\; \Lambda/k - \lambda_i ,
\]
each hyperedge firing at the rate by which the class it omits falls short of its share
$\Lambda/k$.

\begin{remark}[Greedy policies have no freedom here]\label{rem:complete-greedy}
On the complete $k$-uniform hypergraph, all greedy policies coincide. From the empty start,
at most $k-1$ classes are non-empty before each arrival, so that no hyperedge is activable.
An arrival makes one activable only if its class was empty and it brings the number of
non-empty classes to~$k$; these $k$ classes then form the only activable hyperedge. A
greedy policy must activate it, which empties the class of the arrival: at most $k-1$
classes are non-empty again, and every transition was forced. The queue process is thus the
same under every greedy policy. This is the hypergraph
counterpart of the corresponding statement for complete
graphs~\cite[\S3.2]{MP17}, \cite[Proposition~21]{CMV25}.
\end{remark}

Rahme's thesis, whose policies are all non-idling~\cite[Introduction]{Rahme21}, conjectures
precisely this region: in normalized rates $\bar\lambda_i = \lambda_i/\Lambda$,
$k\lambda_i < \Lambda$ reads $\bar\lambda_i < 1/k$, the form in which the conjecture is
stated~\cite[Conclusions and perspectives]{Rahme21}, generalizing the $k = 3$ case
of~\cite[Thm.~1]{RM21}; for $k = 2$ it is the classical region of complete graphs.
\Cref{prop:complete} shows that this is the stabilizable region, over all policies, idling or not. The region is thus an
upper bound on the stability region of every greedy policy, and by \cref{rem:complete-greedy} the greedy form of the
conjecture is a question about one explicitly described chain: whether it is positive
recurrent on the whole region. It is for $k = 2$~\cite{MM16,VGNT21} and $k =
3$~\cite[Thm.~1]{RM21}, and, for every $k$, at uniform rates~\cite{NVGT20}, as explained
next; unequal rates with $k \geq 4$ remain open.

Read as a quantum switch (\cref{sec:intro}), the complete $k$-uniform hypergraph on $n$
classes is the switch serving $k$-partite states to $n$ users, $\lambda_i$ being the
entanglement generation rate of link~$i$, and its unique greedy policy
(\cref{rem:complete-greedy}) is the natural switching rule: measure as soon as $k$
distinct users have a stored qubit.
The region of \cref{prop:complete} coincides with the stability regions known for this
device under its switching rule, so idling cannot enlarge them. For $k = 2$, the condition
$2\lambda_i < \Lambda$ is the stability condition of the bipartite switch with heterogeneous
links, infinite memory and no decoherence~\cite{VGNT21}; for $k = 3$, it is the region
of~\cite[Thm.~1]{RM21}, unequal rates included. At uniform rates,
$k\lambda_i < \Lambda$ holds for every $n \geq k+1$, which is the criterion (more users
than parties) that~\cite{NVGT20} prove for the switching rule itself on identical links.
Under any stable policy every limit point of the mean matching rates satisfies
$k\sum_e \mu_e = \Lambda$ (\cref{prop:rates}, summed over classes), so the switch performs
$\Lambda/k$ measurements per unit time; with a success probability $q$ per measurement,
this is the capacity $q\Lambda/k$ computed in~\cite{VGNT21} for $k = 2$ and
in~\cite{NVGT20} at uniform rates. What \cref{prop:complete} adds is the stabilizable
region for links of unequal rates and $k \geq 4$, a case listed as future work
in~\cite[\S9]{NVGT20}: the switch can be stabilized if and only if no link generates a
fraction $1/k$ or more of the total rate, and whether its switching rule achieves this
region is open.

The region survives blowing each class up into a group of classes, with the per-class
constraint becoming a per-group one.

\begin{corollary}[Blow-ups]\label{cor:blowup}
Let $\V = I_1 \sqcup \dots \sqcup I_p$ with $p \geq r+1$, and let $G$ be the complete
$r$-uniform $p$-partite hypergraph on this partition: the hyperedges are the $r$-subsets
meeting $r$ distinct parts, in one class each. Then $(G,\lambda)$ is stabilizable if and
only if $r\,\lambda(I_j) < \Lambda$ for every part $I_j$.
\end{corollary}

\begin{proof}
\emph{Only if.} Take $y = r\,\mathbf 1_{I_j} - \mathbf 1$. A hyperedge meets $I_j$ in at
most one class, so $\langle y, A_e\rangle = r\bigl(|e \cap I_j| - 1\bigr) \leq 0$, while
$\langle y, \lambda\rangle = r\lambda(I_j) - \Lambda$; \cref{thm:main}(ii) applies.

\emph{If.} Let $y \neq 0$ be such that $\langle y, A_e\rangle \leq 0$ for every
hyperedge~$e$, and let $Y_j = \max_{i \in I_j} y_i$ for each part. Choosing in each part a
class that achieves the maximum shows that $\sum_{j \in J} Y_j \leq 0$ for every set $J$
of $r$ parts: $Y$ satisfies the constraints of condition~(ii) for the complete $r$-uniform
hypergraph on the $p$ parts, with rates $\tilde\lambda_j = \lambda(I_j)$, which is
stabilizable by \cref{prop:complete}. If $Y \neq 0$, \cref{thm:main} thus gives
$\langle Y, \tilde\lambda\rangle < 0$, and
$\langle y, \lambda\rangle \leq \sum_j \lambda(I_j)\, Y_j < 0$. If $Y = 0$, then $y \leq 0$
and $y \neq 0$, so that $\langle y, \lambda\rangle < 0$ as $\lambda > 0$.
\end{proof}

For $r = 3$ this is the region obtained in~\cite[Cor.~4.2]{Rahme21} for complete
$3$-uniform $p$-partite hypergraphs, again for the wider policy class and for every
uniformity $r$.

\paragraph*{Symmetric designs and the Fano plane.}
A \emph{symmetric $2$-$(n,k,t)$ design} has $n$ classes and $n$ hyperedges
(\emph{blocks}) of size~$k$, with $2 \leq k < n$, any two classes lying in exactly $t$
blocks. The Fano plane is the symmetric $2$-$(7,3,1)$ design: seven classes and seven
\emph{lines} of size three, any two classes lying on exactly one line. More generally, the
projective plane of order~$q$ is a symmetric $2$-$(q^2+q+1,\, q+1,\, 1)$ design.

\begin{proposition}[Symmetric designs]\label{prop:fano}
For a symmetric $2$-$(n,k,t)$ design, $A$ is invertible, the unique solution of the
conservation equation is
\[
\mu_\ell \;=\; \frac{\lambda(\ell) - t\Lambda/k}{k-t} \qquad\text{for each block } \ell ,
\]
and $(G,\lambda)$ is stabilizable if and only if every block carries more than the
fraction $t/k$ of the total rate, $\lambda(\ell) > t\Lambda/k$ for all~$\ell$. For the Fano
plane, $\mu_\ell = \tfrac12\bigl(\lambda(\ell) - \Lambda/3\bigr)$, and the condition reads
$\lambda(\ell) > \Lambda/3$ for every line~$\ell$.
\end{proposition}

\begin{proof}
Fix a class~$i$, lying in $r_i$ blocks. Counting the pairs $(j, \ell)$ such that
$j \neq i$ and the block $\ell$ contains $i$ and~$j$ gives $r_i(k-1) = t(n-1)$, so that
$r_i = r$ does not depend on~$i$; counting the incidences gives $nk = nr$, so that
$r = k$. Hence $t(n-1) = k(k-1)$ with $n > k$, and $t < k$. The entry $(AA^\top)_{ij}$
counts the blocks through $i$ and~$j$, so that $AA^\top = (k-t)I + t\,\mathbf 1\mathbf 1^\top$: this
matrix is positive definite, and the square matrix $A$ is
invertible. For the vector $\mu$ of the statement and a class~$i$, the $k$ blocks
through~$i$ contain $i$, and each other class $t$ times in all:
\[
\sum_{\ell \ni i} \mu_\ell
\;=\; \frac{1}{k-t}\Bigl(k\lambda_i + t\,(\Lambda - \lambda_i) - k\,\frac{t\Lambda}{k}\Bigr)
\;=\; \lambda_i .
\]
So $\mu$ is the solution of $A\mu = \lambda$, and, $A$ being surjective,
\cref{thm:main}(iii) reads $\mu > 0$.
\end{proof}

Rahme's thesis conjectures, on the strength of simulations, that the Fano plane is
stabilizable~\cite[Example~4.2 and Conclusions and perspectives]{Rahme21}.
\Cref{prop:fano} settles its analogue for general policies, idling or not, and supplies the
region: uniform rates give $\mu_\ell = \Lambda/21 > 0$ on every line, and the condition
fails as soon as one line is starved. For the non-idling policies of the thesis, this
region is an upper bound, and the conjecture remains open. The thesis also observes that removing one line from
the Fano plane yields a $3$-uniform bipartite hypergraph, which is not stabilizable; in
the present terms, deleting a column from an invertible $7 \times 7$ matrix leaves rank
$6$, so (iii$'$) fails for every $\lambda$.

\section{Stability requirements}\label{sec:requirements}

This section gathers what stability requires of every policy. \Cref{sec:necessity} proves
the necessity half of \cref{thm:main}: outside the region, small queues occupy a vanishing
fraction of the epochs. \Cref{sec:rates} shows that the mean matching rates of a stable
policy solve the conservation equation on a surjective support; this gives a second route
to necessity. It also shows that approaching a solution whose support is not surjective costs mean
queues of order at least $1/\epsilon$. \Cref{rem:wall} shows that, unlike reward-based
analyses, the characterization needs no general-position condition. The proof of
\cref{thm:main} uses only \cref{sec:necessity}.

\subsection{Necessity}\label{sec:necessity}

The necessity half of \cref{thm:main} holds for every policy and takes a few lines. Weigh the
classes by a direction $y$ as in the negation of (ii): arrivals bring non-negative weight on
average and departures never lower the weight. As a result, the queue size is lower bounded by
a random walk with non-negative drift. This is essentially the argument
of~\cite[Prop.~4.1]{MM16} for simple graphs.
\Cref{prop:necessity} extends the idea to
hypergraphs, to keep our results self-contained, and
says more: outside the stability region, small queues occupy a vanishing fraction of the epochs, under every policy.

\begin{proposition}[Necessity for every policy]\label{prop:necessity}
Suppose that condition~(ii) of \cref{thm:main} fails: some $y \neq 0$ has
$\langle y, A_k\rangle \leq 0$ for all $k \in \Ed$ and $\langle y, \lambda\rangle \geq 0$.
Then, under any policy, Markovian or not, randomized or not, even one aware of all future
arrivals, from any deterministic initial state and for every integer $L \geq 0$,
\[
\frac1T \sum_{t<T} \PP\bigl(\|X(t)\|_1 \leq L\bigr) \;\longrightarrow\; 0
\qquad (T \to \infty),
\]
at rate $O(1/T)$ if $\langle y, \lambda\rangle > 0$, and at rate at most $T^{-1/2+o(1)}$
if $\langle y, \lambda\rangle = 0$.
In particular, no policy is stable in the sense of~\labelcref{eq:stable}: $(G,\lambda)$ is not
stabilizable.
\end{proposition}

\begin{proof}
\emph{Pathwise bound.} Since $\lambda > 0$, $y \neq 0$ and
$\langle y, \lambda\rangle \geq 0$, at least one coordinate of~$y$ is positive; rescaling~$y$,
we assume that $\max_i y_i = 1$, so that $\langle y, x\rangle \leq \|x\|_1$ for every
$x \in \N^n$. Let $Z_t = \langle y, N(t)\rangle$. This walk depends on the arrivals alone,
not the policy: its increments $y_{a(u)}$ are i.i.d., bounded by $\|y\|_\infty$, with mean
$\theta = \langle y, \bar\lambda\rangle \geq 0$. By~\labelcref{eq:balance},
$\langle y, X(t)\rangle - Z_t = \langle y, X(0)\rangle - \sum_k \langle y, A_k\rangle
M_k(t)$ is nondecreasing in~$t$, whatever the policy. Hence, along every trajectory,
\[
\|X(t)\|_1 \;\geq\; \langle y, X(t)\rangle \;\geq\; Z_t + \langle y, X(0)\rangle ,
\]
and the event $\{\|X(t)\|_1 \leq L\}$ implies $Z_t \leq c := L - \langle y, X(0)\rangle$.

\emph{Positive drift.} If $\theta > 0$, Hoeffding's inequality gives
$\PP(Z_t \leq c) \leq \exp\bigl(-(\theta t - c)^2 / (2\|y\|_\infty^2\, t)\bigr)$ for
$t > c/\theta$. Hence $\sum_{t \geq 0} \PP(\|X(t)\|_1 \leq L) < \infty$, and the Ces\`aro
average is $O(1/T)$.

\emph{Zero drift.} If $\theta = 0$, the variance $\sum_i \bar\lambda_i y_i^2$ of the
increments is positive, and the central limit theorem gives $\PP(Z_t \leq c) \to \tfrac12$.
The Ces\`aro average is thus at most $\tfrac12 + o(1)$ for every~$L$: large queues occupy
at least half of the epochs in the long run, against~\labelcref{eq:stable}, and this is
all that \cref{thm:main} needs. For the limit~$0$, group the epochs $t < T$ with $\|X(t)\|_1 \leq L$ according to the value
of $X(t)$, which ranges over at most $m_L = \binom{n+L}{n}$ vectors. Within a group,
$\langle y, X(t)\rangle$ is constant, so $Z$ is non-increasing along the group. Hence,
pathwise,
\[
\#\bigl\{t < T : \|X(t)\|_1 \leq L\bigr\} \;\leq\; m_L\,\ell_T,
\]
where $\ell_T$ is the length of a longest non-increasing subsequence of
$Z_0, \dots, Z_{T-1}$, that is, of a longest weakly increasing subsequence of the centered
walk $-Z$. By~\cite[Cor.~1.2]{ABP17}, $\E\,\ell_T \leq T^{1/2 + o(1)}$.
\end{proof}

The proof is pathwise: the activations never decrease
$\langle y, X\rangle$, so along the epochs where $X$ takes a given value the arrival walk
$\langle y, N\rangle$ can only go down, and a centered random walk has no non-increasing
subsequence of linear length. The exponent $\tfrac12$ is attained: with the single
hyperedge $\{1,2\}$ and $\lambda = (1,1)$, every policy has
$\|X(t)\|_1 \geq |X_1(t) - X_2(t)|$, the absolute value of a simple random walk, with
equality for the policy that matches whenever possible.

\subsection{Matching rates of stable models}\label{sec:rates}

In~\cite{CMV25}, the long-run matching rates of a positive recurrent model solve the
conservation equation and, by~\cite[Prop.~10]{CMV25}, never sit at a degenerate vertex of
the polytope. The argument rests on the Markov structure, under which
a hyperedge with zero rate is never activated at all. Without it, a policy may activate a
hyperedge $o(T)$ times: too rarely to show in the rates, but, one might fear, often
enough to steer the queues. This subsection shows that it cannot. Every cluster point of
the mean matching rates of a stable policy, Markovian or not, clairvoyant or not, obeys the
same two constraints.
The proof is a quantitative version of \cref{prop:necessity}: along a direction $y$ in
which the arrivals have no drift, the arrival walk $\langle y, N\rangle$ moves by
order~$L$ over any lag of order~$L^2$; a policy whose queues are small at both ends of
such a lag must have cancelled that move by activating hyperedges that are not orthogonal
to~$y$, and over a horizon $T$ this costs a number of activations of order $T/L$.

Throughout, the system starts empty, and we use the notation of
\cref{sec:model}: $a(t)$, $s(t)$, $N(t)$, $M(t)$ and the balance
equation~\labelcref{eq:balance}. Nothing is assumed on the policy beyond the
availability constraint. Finally, as each activation removes at least one item and at
most $T$ items have arrived,
\begin{equation}\label{eq:M-bound}
M_k(T) \leq T \quad\text{for every } k \in \Ed \text{ and } T \in \N .
\end{equation}

\begin{lemma}[Driftless directions are expensive]\label{lem:driftless}
Let $y \in \R^n$ with $\|y\|_\infty = 1$ and $\langle y, \lambda\rangle = 0$, let
$\sigma^2 = \sum_i \lambda_i y_i^2$, which is positive, let
\[
V_y(T) = \sum_{k \in \Ed} |\langle y, A_k\rangle|\, M_k(T)
\]
be the $y$-weighted number of activations before epoch~$T$, and let
$\Delta_h = \langle y, N(h)\rangle$ be the displacement of the arrival walk over $h$
epochs. Under any policy, for all integers $T \geq h \geq 1$,
\begin{equation}\label{eq:window-mean}
h\,\E\,V_y(T) \;\geq\; (T-h)\,\E|\Delta_h| - 2\sum_{t<T} \E\,\|X(t)\|_1 .
\end{equation}
Moreover, if the model is stable, then $\liminf_{T \to \infty} \E\,V_y(T)/T > 0$: in
expectation, a stable policy activates the hyperedges $k$ with
$\langle y, A_k\rangle \neq 0$ at a positive total rate.
\end{lemma}

\begin{proof}
Write $W(t) = \langle y, X(t)\rangle$, $D(t) = \langle y, N(t)\rangle$ and
$C(t) = \sum_k \langle y, A_k\rangle M_k(t)$, so that $W = D - C$
by~\labelcref{eq:balance}. As $\|y\|_\infty = 1$, $|W(t)| \leq \|X(t)\|_1$; and
$|C(t+h) - C(t)| \leq V_y(t+h) - V_y(t)$, as each increment of $C$ is bounded in absolute
value by the corresponding increment of~$V_y$. Hence
$D(t+h) - D(t) = W(t+h) - W(t) + C(t+h) - C(t)$ gives, for all $t \geq 0$ and $h \geq 1$,
\begin{equation}\label{eq:window}
|D(t+h) - D(t)| \;\leq\; \|X(t)\|_1 + \|X(t+h)\|_1 + V_y(t+h) - V_y(t) .
\end{equation}
Sum over $t < T - h$. Each epoch of $[0, T)$ occurs at most twice, as $t$ or as $t+h$,
and each increment $V_y(s+1) - V_y(s)$ belongs to at most $h$ of the windows $[t, t+h)$,
so that $\sum_{t<T-h} (V_y(t+h) - V_y(t)) \leq h\,V_y(T)$. As $D(t+h) - D(t)$ is
distributed as $\Delta_h$ for every~$t$, whatever the policy, taking expectations
gives~\labelcref{eq:window-mean}.

For the last claim, let $L > 0$. If the left-hand side of~\labelcref{eq:window} is at
least $3L$, then $\|X(t)\|_1 \geq L$, or $\|X(t+h)\|_1 \geq L$, or
$V_y(t+h) - V_y(t) \geq L$, and the last case occurs for at most $h\,V_y(T)/L$ values of
$t < T - h$. Counting the three cases and taking expectations,
\begin{equation}\label{eq:pairs}
h\,\E\,V_y(T) \;\geq\; L\Bigl((T-h)\,\PP\bigl(|\Delta_h| \geq 3L\bigr)
- 2\sum_{t<T} \PP\bigl(\|X(t)\|_1 \geq L\bigr)\Bigr).
\end{equation}
Let $\varepsilon_L = \limsup_{T} \frac1T \sum_{t<T} \PP(\|X(t)\|_1 \geq L)$, which vanishes
as $L \to \infty$ by~\labelcref{eq:stable}. Dividing by $hT$,
$\liminf_T \E\,V_y(T)/T \geq \frac{L}{h}\bigl(\PP(|\Delta_h| \geq 3L) - 2\varepsilon_L\bigr)$.
Fix $L$ such that $\varepsilon_L \leq \tfrac18$. As $\Delta_h$ is a sum of $h$ i.i.d.\
centered increments, of positive variance $\sigma^2/\Lambda$, the central limit theorem
gives $\PP(|\Delta_h| \geq 3L) \to 1$ as $h \to \infty$: for $h$ large, the right-hand
side is at least $L/(4h) > 0$.
\end{proof}

The same comparison between the arrival walk and the activations quantifies a trade-off
between the $y$-weighted activation rate and the size of the queues.

\begin{corollary}[Queue--activation trade-off]\label{cor:tradeoff}
Let $y$, $\sigma^2$ and $V_y$ be as in \cref{lem:driftless}, and let $\Phi$ be any policy
whose $y$-weighted activation rate is at most $r > 0$, that is,
$\limsup_{T \to \infty} \frac{\Lambda}{T}\,\E\,V_y(T) \leq r$. Then the mean queue length is
at least of order $\sigma^2/r$:
\[
\liminf_{T \to \infty} \frac1T \sum_{t<T} \E\,\|X(t)\|_1 \;\geq\; \frac{\sigma^2}{32\,r} - \frac18 .
\]
\end{corollary}

\begin{proof}
Let $q$ be the limit inferior on the left, and $h \geq 1$ an integer.
Dividing~\labelcref{eq:window-mean} by~$T$ and taking the limit inferior in~$T$,
$\E|\Delta_h| \leq 2q + hr/\Lambda$, as $\limsup_T h\,\E\,V_y(T)/T \leq hr/\Lambda$.
Now $\Delta_h$ is a sum of $h$ i.i.d.\ increments $\xi$, centered, bounded by~$1$, of
variance $\sigma^2/\Lambda$, so that $\Delta_h$ has variance $v = h\sigma^2/\Lambda$.
Expanding the fourth moment, with $\E\xi^4 \leq \E\xi^2$,
$\E\Delta_h^4 = h\,\E\xi^4 + 3h(h-1)\,(\E\xi^2)^2 \leq v + 3v^2$, which is at most $4v^2$
if $v \geq 1$. H\"older's inequality,
$\E\Delta_h^2 \leq (\E|\Delta_h|)^{2/3} (\E\Delta_h^4)^{1/3}$, then gives
$\E|\Delta_h| \geq v^{3/2}\,(4v^2)^{-1/2} = \tfrac12 \sqrt v$. As
$hr/\Lambda = (r/\sigma^2)\,v$,
\[
2q \;\geq\; \tfrac12 \sqrt v - \frac{r}{\sigma^2}\,v
\qquad\text{for every integer $h$ such that } v = h\sigma^2/\Lambda \geq 1 .
\]
If $r \geq \sigma^2/4$, the claimed bound is nonpositive and there is nothing to prove.
Otherwise, take $h = \lceil \Lambda\sigma^2/(16r^2) \rceil$. As $\sigma^2/\Lambda \leq 1$,
we have $\sigma^4/(16r^2) \leq v \leq \sigma^4/(16r^2) + 1$, so that, as $r < \sigma^2/4$,
$v > 1$ and
$2q \geq \sigma^2/(8r) - \sigma^2/(16r) - r/\sigma^2 \geq \sigma^2/(16r) - \tfrac14$,
which is the claim.
\end{proof}

The order $1/r$ cannot be improved. Take two classes with
$\lambda = (1,1)$, the pair $\{1,2\}$ and the two mono-edges, and $y = (1,-1)$, so that
$\sigma^2 = 2$. The policy that matches the arriving item with a queued item of the
other class whenever possible, queues it otherwise, unless its own class already holds
$\ell$ items, in which case it discards it through its mono-edge, keeps
$\|X(t)\|_1 = |X_1(t) - X_2(t)| \leq \ell$ at all times; $X_1 - X_2$ is then a simple
random walk on $\{-\ell, \dots, \ell\}$ held at the endpoints, uniform in stationarity, and a
mono-edge fires on half of the arrivals at an endpoint, so that
$\frac{\Lambda}{T}\,\E\,V_y(T) \to 2/(2\ell+1)$: the rate is $r = \sigma^2/(2\ell+1)$, and the
queues never exceed $\ell < \sigma^2/(2r)$.

The trade-off also holds in frequency, with cruder constants. For a policy as in
\cref{cor:tradeoff} and an integer $L \geq 1$, take
$h = \lceil 18\Lambda L^2/\sigma^2 \rceil$ in the inequality~\labelcref{eq:pairs} of the
proof of \cref{lem:driftless}. The variance $v$ of $\Delta_h$ is then at least $18L^2$,
and the Paley--Zygmund inequality, with the fourth moment bounded in the proof of
\cref{cor:tradeoff}, gives
$\PP(|\Delta_h| \geq 3L) \geq \PP(\Delta_h^2 \geq \tfrac12 v) \geq \tfrac14\, v^2 / \E\Delta_h^4
\geq \tfrac1{16}$. As $h \leq 19\Lambda L^2/\sigma^2$, dividing~\labelcref{eq:pairs} by
$LT$ shows that the queues hold at least $L$ items at an expected fraction of the epochs
whose limit inferior is at least $\frac1{32} - \frac{19\,Lr}{2\sigma^2}$, hence at least
$\frac1{64}$ as long as $L \leq \sigma^2/(608\,r)$. Large queues are thus not rare
excursions.

\begin{proposition}[Matching rates of stable models]\label{prop:rates}
Let $(G, \lambda, \Phi)$ be a matching model, and let
\[
\mu^{(T)} = \frac{\Lambda}{T}\, \E\, M(T) \in [0, \Lambda]^m
\]
be its mean matching rates up to epoch~$T$, per unit of time (arrival epochs occur at
rate~$\Lambda$).
\begin{enumerate}
\item[(a)] If the policy $\Phi$ is Markovian and the model positive recurrent, then
$\mu^{(T)}$ converges. Its limit~$\mu$ is given by $\mu_k = \Lambda \sum_z \pi(z)\, m_k(z)$, where $\pi$ is the stationary law
of the embedded chain on the class of $\emptystate$ and $m_k(z)$ the mean number of
activations of~$k$ decided from state~$z$; and $\mu$ satisfies $A\mu = \lambda$, with
a surjective support.
\item[(b)] In general, under any stable policy~$\Phi$, every cluster point~$\mu$ of
$(\mu^{(T)})_{T \geq 1}$ satisfies $A\mu = \lambda$, and the restriction of $A$ to the columns in $\supp(\mu)$ is
surjective.
\end{enumerate}
In both cases, $\mu$ is a witness of condition~(iii$'$) of \cref{thm:main}, and $\lambda$
lies in the interior of the cone generated by the hyperedges of $\supp(\mu)$.
\end{proposition}

\begin{proof}
We prove (b), then (a). The sequence lies in $[0, \Lambda]^m$ by~\labelcref{eq:M-bound},
so cluster points exist.

\emph{Conservation.} Taking expectations in~\labelcref{eq:balance} and dividing by~$T$,
$A\,\E M(T)/T = \bar\lambda - \E X(T)/T$, and the last term vanishes as $T \to \infty$ by
\cref{lem:sublinear}. Hence $A\mu = \Lambda \bar\lambda = \lambda$ for every cluster point.

\emph{Support.} Let $J = \supp(\mu)$, let $A_J$ be the restriction of $A$ to the columns
in~$J$, and suppose that $A_J$ is not surjective: some
$y$ with $\|y\|_\infty = 1$ satisfies $\langle y, A_k\rangle = 0$ for all $k \in J$. Then
$\langle y, \lambda\rangle = \langle y, A\mu\rangle = \sum_{k \in J} \mu_k \langle y, A_k\rangle = 0$,
so \cref{lem:driftless} applies: $\liminf_T \E\,V_y(T)/T > 0$. But along a subsequence
$T_j$ with $\mu^{(T_j)} \to \mu$,
\[
\frac{\E\, V_y(T_j)}{T_j} = \sum_{k \in \Ed} |\langle y, A_k\rangle|\,\frac{\E\,M_k(T_j)}{T_j}
\;\longrightarrow\; \frac1\Lambda \sum_{k \in \Ed} |\langle y, A_k\rangle|\, \mu_k = 0,
\]
each term vanishing because $\mu_k = 0$ for $k \notin J$ and $\langle y, A_k\rangle = 0$
for $k \in J$. This contradiction shows that $A_J$ is surjective; as $\mu$ restricted to
$J$ is positive, \cref{lem:geometry}, (b)\,$\Rightarrow$\,(a), applied to $A_J$ places
$\lambda$ in the interior of $\cone(A_J)$. This proves (b).

\emph{Positive recurrent case.} By the countable Markov structure of
\cref{sec:stability_notions}, the law of $s(t)$ given $S(t) = z$ is the same at every
epoch, so $m_k(z)$ is well defined and $\E[s_k(t) \mid S(t)] = m_k(S(t))$. Hence
$\E\,M_k(T)/T = \sum_z p_T(z)\, m_k(z)$, where $p_T$ is the Ces\`aro average
over $t < T$ of the laws of $S(t)$; by~\labelcref{eq:M-bound}, this sum is at most~$1$.
As noted in \cref{sec:stability_notions}, $p_T \to \pi$ pointwise, so that
$\sum_z \pi(z)\, m_k(z) \leq 1$ by Fatou's lemma. Moreover
$p_T \leq \pi/\pi(\emptystate)$: with $P$ the transition kernel of the embedded chain on
the class of~$\emptystate$, $\pi = \pi P^t \geq \pi(\emptystate)\,\delta_\emptystate P^t$
for every~$t$, and $\delta_\emptystate P^t$ is the law of $S(t)$. Dominated convergence
then gives $\E\,M_k(T)/T \to \sum_z \pi(z)\, m_k(z)$: the sequence
converges, and its limit, being its only cluster point, satisfies both constraints by~(b),
a positive recurrent model being stable (\cref{sec:stability_notions}).
\end{proof}

\begin{corollary}[The cost of approaching a non-surjective support]\label{cor:degenerate}
Assume that $A$ is surjective, and let $\mu^\star \in \R^m_{\geq 0}$ be a solution of the
conservation equation, $A\mu^\star = \lambda$, such that the restriction of $A$ to the
columns in $J = \supp(\mu^\star)$ is not surjective. Let $y$, with $\|y\|_\infty = 1$, be
orthogonal to these columns, $\langle y, A_k\rangle = 0$ for all $k \in J$, let
$\sigma^2 = \sum_i \lambda_i y_i^2$ and $\kappa = \max_{k \in \Ed} |\langle y, A_k\rangle|$,
let $\epsilon > 0$, and let $\Phi$ be any policy whose mean matching rates
$\mu^{(T)} = \frac{\Lambda}{T}\,\E\,M(T)$ satisfy
$\limsup_{T \to \infty} \|\mu^{(T)} - \mu^\star\|_1 \leq \epsilon$.
Then
\[
\liminf_{T \to \infty} \frac1T \sum_{t<T} \E\,\|X(t)\|_1
\;\geq\; \frac{\sigma^2}{32\,\kappa\,\epsilon} - \frac18 .
\]
If moreover $\Phi$ is Markovian and the model positive recurrent, with stationary
law~$\pi$, then
\[
\E_\pi \|X\|_1 \;\geq\; \frac{\sigma^2}{32\,\kappa\,\epsilon} - \frac18 .
\]
\end{corollary}

\begin{proof}
As $\langle y, \lambda\rangle = \langle y, A\mu^\star\rangle
= \sum_{k \in J} \mu^\star_k \langle y, A_k\rangle = 0$, \cref{cor:tradeoff} applies
to~$y$; and $\kappa > 0$, as $A$ is surjective. The hyperedges
of~$J$ do not contribute to $V_y$, and $\mu^\star_k = 0$ for $k \notin J$, so that
\[
\frac{\Lambda}{T}\,\E\,V_y(T)
= \sum_{k \notin J} |\langle y, A_k\rangle|\, \mu^{(T)}_k
\;\leq\; \kappa\, \|\mu^{(T)} - \mu^\star\|_1 ,
\]
whose $\limsup$ is at most $\kappa\epsilon$. \Cref{cor:tradeoff} with $r = \kappa\epsilon$
gives the first bound.

In the positive recurrent case, assume that $\E_\pi\|X\|_1 < \infty$ (otherwise there is
nothing to prove). As in the proof of \cref{prop:rates}, the Ces\`aro averages $p_T$ of
the laws of $S(t)$ converge to~$\pi$ pointwise and are dominated by
$\pi/\pi(\emptystate)$, so that dominated convergence gives
$\frac1T \sum_{t<T} \E\,\|X(t)\|_1 \to \E_\pi\|X\|_1$: the second bound is the first.
\end{proof}

The corollary applies in particular to every vertex of the polytope $\{\mu \geq 0 : A\mu = \lambda\}$
whose support has fewer than $n$ hyperedges. On a simple graph every column of~$A$ has two
unit entries, so that $\kappa \leq 2$, and
the vertices concerned are the \emph{injective-only} vertices of~\cite{CMV25}, which no
positive recurrent policy reaches~\cite[Prop.~10]{CMV25} and, by \cref{prop:rates}, no
stable policy either. For a positive recurrent policy whose matching
rates are within~$\epsilon$ of such a vertex, \cref{cor:degenerate} gives a mean queue length
of order at least $1/\epsilon$. This is the first-moment lower bound conjectured
after~\cite[Thm.~2]{CMV25}; it matches the $O(1/\epsilon)$ achieved by
$\epsilon$-filtering~\cite[Thm.~1]{CMV25}. As $\|x\|_2^2 \geq \|x\|_1^2/n$, the second moment
$\E_\pi\|X\|_2^2$ is at least $(\E_\pi\|X\|_1)^2/n$, of order $1/\epsilon^2$, which improves
the order $1/\epsilon$ of~\cite[Thm.~2]{CMV25}.
In the model with mono-edges, the threshold policy described after \cref{cor:tradeoff}
attains both orders: with threshold~$\ell$, its matching rates are at distance
$\epsilon = 3/(2\ell+1)$ from the vertex carried by the pair alone, its mean queue length
is $\ell(\ell+1)/(2\ell+1)$, of order $3/(4\epsilon)$, twelve times the bound
$1/(16\epsilon)$ of the corollary, and its second moment is
$\ell(\ell+1)/3$, of order $3/(4\epsilon^2)$.

\begin{remark}[Scope of the two necessity results]\label{rem:necessity-scope}
\Cref{prop:necessity} holds for every policy: greedy or idling, matching on arrival or not,
with any tie-breaking rule, Markovian or not, even one aware of future arrivals. Its proof
uses two facts only: the arrivals are exogenous and i.i.d., and the queues obey the
balance~\labelcref{eq:balance}. Condition~(ii) is therefore necessary under every
stability notion that makes large queues rare in time average, which is what
\labelcref{eq:stable} asks and what positive recurrence, as noted in
\cref{sec:stability_notions}, or the drift-to-a-finite-set notion of~\cite{AKRSVW04}
imply: the equivalence of \cref{thm:main} does not depend on the notion, and it even holds
for a notion of another kind, based on clearing times (\cref{rem:nb26}).
\Cref{prop:rates} says more about stable policies: whatever rates they approach have a
surjective support. In particular, the degenerate vertices of the polytope
$\{\mu \geq 0 : A\mu = \lambda\}$, whose support has fewer than $n$ hyperedges, are out of
reach of every stable policy, in expectation and as an almost sure limit: no such vertex
is a cluster point of $\mu^{(T)}$, hence none is an almost sure limit of $\Lambda M(T)/T$
either, by bounded convergence and~\labelcref{eq:M-bound}. Pathwise cluster points are not
excluded: in the two-class example that follows \cref{cor:tradeoff}, a policy that cuts
time into longer and longer blocks of epochs and, on rare random blocks, matches pairs
only, while it runs the threshold policy on the others, is stable, yet $\Lambda M(T)/T$
comes arbitrarily close to the vertex carried by the pair infinitely often, with
probability one, while $\mu^{(T)}$ converges elsewhere. The proposition also gives a second
proof of necessity, for every policy: (i) implies (iii$'$) directly, any cluster point of
the mean rates of a stable policy being a witness. \Cref{cor:degenerate} quantifies the
cost of approaching such a vertex: the closer the mean rates, the larger the mean queues,
under every policy.
\end{remark}

\subsection{General position is not necessary}\label{rem:wall}

Reward-based analyses rest on general-position conditions on the static planning linear
program: uniqueness and non-degeneracy of its optimum~\cite{KAG24}, or persistence of its
optimal basis under small perturbations of~$\lambda$~\cite{G24}. Such conditions hold
generically: for almost every~$\lambda$, every vertex of $\{\mu \geq 0 : A\mu = \lambda\}$
is non-degenerate, and for almost every reward the optimum is a single
vertex~\cite[Props.~7 and~9]{CMV25}, two statements whose proofs carry over verbatim to
hypergraphs (see also~\cite{G24}). Genericity, however, is a weaker reassurance than it sounds,
as~\cite[\S4.2.2]{CMV25} already warn. The regularities of an application,
degree-proportional rates for instance, confine $\lambda$ to lower-dimensional sets on which
degeneracy is often the rule. For example, under blood-type compatibility, every mirror-invariant distribution
of pair types, one that is unchanged when the donor and recipient types of a pair are
swapped, makes the problem unstabilizable~\cite[Thm.~2]{Mat27}: (ii) fails through a
certificate with $\langle y, \lambda\rangle = 0$, an equality that a $\lambda$ drawn at
random satisfies with probability zero.

\Cref{thm:main} needs no such condition. For example, let $n = 3$ and let
\[
A = \begin{pmatrix}
	0 & 1 & 1 & 0\\ 
	0 & 0 & 1 & 1\\ 
	1 & 1 & 1 & 1\\ 
	\end{pmatrix}\text{, with }\lambda = (1,1,2)\text{.}
\]
Then $A$ is surjective and the positive solutions of $A\mu = \lambda$ form the segment
$\mu(t) = (t,\, 1-t,\, t,\, 1-t)$, $t \in (0,1)$, whose two vertices have supports
$\{1,3\}$ and $\{2,4\}$, each of rank~$2$: no reward whose optimum matches every arrival
satisfies either condition, yet $(G,\lambda)$ is stabilizable, as $\mu(1/2)=(1/2,1/2,1/2,1/2)$ satisfies (iii).
Moreover, unless the reward is constant on the segment, the optimum is out of reach of any stable policy, Markovian or not (see \cref{rem:necessity-scope}).
Such instances already occur on simple graphs, for instance on the diamond with degree-proportional rates~\cite[Example~6]{CMV25}, but~\cite{MM16} covers that case.

\section{Sufficiency: a maximally stable policy}\label{sec:sufficiency}

The sufficiency proof uses the mechanism of EGPD~\cite{NS19}: a signed virtual queue
driven by MaxWeight~\cite{TE92}, decoupled from the physical feasibility constraints, together
with the completion mechanism that slaves the physical system to it~\cite[Prop.~3]{NS19}. In
both its original and rewardless forms, EGPD is $\lambda$-oblivious, and its drift
analysis~\cite[\S3.5]{NS19}, via~\cite[\S4.9]{stolyar05}, already delivers stability in the
sense standard in that literature: a drift towards a finite set reached from any initial
state, hence a stationary regime~\cite{AKRSVW04}. We take that over. The sufficiency half
of \cref{thm:main} is EGPD applied to hypergraph matching and we claim no novelty for it;
the drift computation is reproduced (\cref{lem:foster}) only so that the paper is
self-contained, and the absence of boundary reflection makes it shorter than for
nonnegative queues. The completion mechanism (\cref{sec:fullstate}) is a variant of that
of~\cite[Sec.~3.3--3.4]{NS19}: planned matchings wait in a first-come-first-served queue
until the items present can complete them, and no item is reserved. Their scan of this
queue precedes the arrival and stops at the first planned matching that can be completed;
ours follows the arrival and executes every such planned matching, so that none is left at
the beginning of an epoch. The bounds of \cref{lem:pathwise} are then those of their
Proposition~3; we prove them for completeness, the two copies planned at an epoch being
completed separately, which keeps the constant $a_{\max} - 1$, with
$a_{\max} = \max_k \|A_k\|_1$, where a composite matching would roughly double it. One
convention is fixed once and for all: at each epoch the decision is taken from the virtual
queue \emph{before} the arrival lands. Nazari and Stolyar do not characterize
stabilizability, but remark informally~\cite[\S3.6.2]{NS19} that their Assumption~5 is
essentially stabilizability together with a condition that the queues ``can be moved in
any direction''; the inequality~\labelcref{eq:score-gap} in the proof of \cref{lem:foster}
is the dual form of that assumption, derived from $\lambda\in\interior\cone(A)$, and
\cref{prop:necessity} shows that stabilizability forces the latter.

Fix $(G,\lambda)$ with $\lambda\in\interior\cone(A)$; we build one policy, \vqml{}, and
show that it stabilizes the physical system. Once the policy is defined
(\cref{sec:vqml-policy,sec:fullstate}), the argument has two stages. The first is purely about an
auxiliary chain: \vqml{} maintains a \emph{signed} virtual queue $Q\in\Z^n$ (coordinates
may go negative) driven by MaxWeight\footnote{Strictly speaking, \vqml{} is a virtual-queue \emph{MaxWeight} policy; its name, introduced with the policy in~\cite{CMV25}, echoes Match-the-Longest (\ml{}), its greedy counterpart.}, decoupled from the feasibility
constraints of the physical queues. On this chain alone the interior hypothesis
$\lambda\in\interior\cone(A)$ gives a uniform inward drift in every direction,
making $\|Q\|_2^2$ a Foster--Lyapunov function with negative drift
outside a finite set (\cref{lem:foster}), so that the time average of $\E\|Q(t)\|_1$ is
bounded (\cref{sec:virtual}). The second stage is about the real system: following the
completion mechanism pathwise shows that the physical quantities are governed by $Q$, the
number of planned matchings being at most $\|Q^-\|_1$, with $Q^-$ the negative part of $Q$,
and the number of items at most
$a_{\max}\|Q\|_1$ (\cref{lem:pathwise}); the bound on $Q$ then transfers to the physical
queues (\cref{prop:vqml-stable}). This proves (iii)\,$\Rightarrow$\,(i) together with the
explicit mean-queue bound and, since \vqml{} depends only on $G$, maximal stability.
Appendix~\ref{app:pr} adds the positive recurrence of the full chain,
(iii)\,$\Rightarrow$\,(i$'$): there, and only there, the origin must be shown accessible
from every state of the virtual chain, which is where the tie-breaking rule enters.
\Cref{sec:putting-things-together} chains the implications into the proof of
\cref{thm:main}, and \cref{rem:nb26} compares it with the concurrent work~\cite{NB26}.

\subsection{The policy \vqml{}}\label{sec:vqml-policy}
In words, \vqml{} plans matchings ahead of the items. It keeps a signed ledger
$Q \in \Z^n$: $Q_i > 0$ counts the class-$i$ items present in excess of what the planned
matchings need, $Q_i < 0$ counts the class-$i$ items that the planned matchings still
miss. At each arrival the policy first plans, from the ledger alone, two copies of
the hyperedge of largest score, the total of the ledger over its classes, with
multiplicity, if this score is positive, and nothing otherwise; the arriving item is then recorded in the ledger, and the planned
matchings that the items present can complete are executed, oldest first, their items
leaving the system. On a single hyperedge
$\{1,2,3\}$, for instance, a first arrival of class~$1$ sets $Q = (1,0,0)$; at the next
epoch the hyperedge has score~$1$, two copies are planned, $Q$ becomes $(-1,-2,-2)$ before
the arrival is recorded, five items are missing to complete the two copies,
and the policy plans nothing further as long as the score $Q_1+Q_2+Q_3$ is not
positive. The ledger never consults the physical queues, and the physical system
follows it: the planned matchings number at most $\|Q^-\|_1$, and the items at most
$a_{\max}\|Q\|_1$ (\cref{lem:pathwise}).

Following Nazari and Stolyar~\cite{NS19}, run a \emph{virtual} system alongside the
physical one. The virtual queue vector $Q(t) \in \Z^n$ (signed, no reflection) starts
at $Q(0) = 0$. At each arrival epoch $t$ (arrival vector $a(t) \in \{e_1, \dots,
e_n\}$), the controller plans a multiset of at most two hyperedges\footnote{Nazari and Stolyar allow several matchings to be chosen per epoch and present the case of
	one~\cite[\S3.1]{NS19}; here two planned matchings, chosen jointly from $Q(t)$, suffice
	for every instance, and \cref{rem:budget} shows that one does not.}, chosen by
\emph{signed MaxWeight}:
\begin{equation}\label{eq:vqml}
s(t) \in \operatorname*{arg\,max}_{s \in \N^m,\ \|s\|_1 \leq 2}
\ \langle Q(t), A s \rangle ,
\qquad
Q(t+1) = Q(t) + a(t) - A\,s(t).
\end{equation}
Concretely, $\langle Q, A_k\rangle$ is the \emph{score} of hyperedge $k$, the sum of the
virtual queues of its classes, counted with multiplicity, and the maximizers of \cref{eq:vqml} are explicit: when some
score is positive, they are the pairs of best-scoring hyperedges; when no score is
positive, planning nothing is a maximizer, as every $s \geq 0$ then scores at most $0$.
Ties may be broken in any way that does not anticipate the arrivals, using the past or
randomness independent of the arrivals: the main text needs no
tie-breaking rule, as the drift bound of \cref{lem:foster} holds for any such maximizer and the
pathwise bounds of \cref{lem:pathwise} for any decision sequence. Appendix~\ref{app:pr},
where \vqml{} must be Markovian and the origin accessible, fixes a canonical rule. The
policy uses no knowledge of $\lambda$ and no parameter: scaling $Q$ by any positive factor
changes neither the signs of the scores nor their order.

\subsection{The physical mechanism}\label{sec:fullstate}
Write $a_{\max} = \max_k \|A_k\|_1$, and $q^+ = \max(q, 0)$, $q^- = \max(-q, 0)$
componentwise, so that $q = q^+ - q^-$ and $\|q\|_1 = \|q^+\|_1 + \|q^-\|_1$.

The physical system follows the virtual queue through a variant of the completion
mechanism of EGPD~\cite[Sec.~3.3--3.4]{NS19}. Each copy of a hyperedge~$k$ selected in \cref{eq:vqml} is
a \emph{planned matching} of type~$k$. Planned matchings wait in the \emph{backlog}, the
list of their types, oldest first: they are ordered by the epoch at which they are planned
and, within an epoch, by hyperedge index. A planned matching of type~$k$ can be
\emph{executed} when the items present cover it, that is, when at least $A_{i,k}$ items of
class~$i$ are present, for every class~$i$: these items then leave the system, and the
planned matching leaves the backlog. Within epoch~$t$, the decision $s(t)$ is taken from
$Q(t)$ and its planned matchings join the backlog, then the item $a(t)$ arrives, and the
backlog is scanned once, oldest first, each planned matching that the items still present
cover being executed. No item is assigned to a planned matching before its execution. The
executions are the physical activations of \cref{sec:policies}, which $M(t)$ counts in the
balance equation~\labelcref{eq:balance}; in this section and in Appendix~\ref{app:pr},
$s(t)$ counts the matchings planned at epoch~$t$, not the activations.

Let $Y(t)$ be the backlog; for a list~$Y$ of hyperedges, $|Y|$ is its length and $|Y|_k$
the number of its entries equal to~$k$. The state of the \vqml-controlled system is
\[
S(t) = \bigl( X(t),\, Y(t) \bigr), \qquad\text{with}\quad \emptystate = \bigl(0,\, ()\bigr),
\]
where $()$ is the empty list. Both components are taken at the beginning of epoch~$t$,
before the decision and the arrival, as in~\labelcref{eq:balance}, and $S(t+1)$ is a
function of $S(t)$, $s(t)$ and $a(t)$: when the decision is a function of $Q(t)$ alone, as
in Appendix~\ref{app:pr}, the backlog is the internal state of \vqml{} in the sense of
\cref{sec:stability_notions}. The state does not record the virtual queue
explicitly, but the virtual queue can be read from it: it is the signed queue vector that
executing the whole backlog at once would leave (\cref{lem:pathwise}(a)), so that $Q_i^+$
is the surplus of class-$i$ items over what the planned matchings need, and $Q_i^-$ their
shortage. The virtual queue, for its part, is \emph{autonomous}: by \cref{eq:vqml}, it is
driven by the arrivals and the decisions alone.

\begin{lemma}[Pathwise structure of the full state]\label{lem:pathwise}
Whatever the decisions $s(t) \in \N^m$, with $Q$ defined from them by \cref{eq:vqml} and
$Q(0) = 0$, every trajectory from the empty start satisfies, for all $t \in \N$:
\begin{enumerate}
\item[(a)] \emph{(virtual queue)} $Q(t) = X(t) - \sum_{k \in \Ed} |Y(t)|_k\, A_k$;
\item[(b)] \emph{(backlog bound)} $|Y(t)| \leq \|Q^-(t)\|_1$;
\item[(c)] \emph{(queue bound)} $\|X(t)\|_1 \leq \|Q^+(t)\|_1 +
(a_{\max} - 1)\,\|Q^-(t)\|_1 \leq a_{\max}\|Q(t)\|_1$;
\item[(d)] \emph{(regeneration)} $S(t) = \emptystate$ if and only if $Q(t) = 0$.
\end{enumerate}
\end{lemma}

\begin{proof}
(a) Of the $\sum_{u<t} s_k(u)$ matchings of type~$k$ planned before epoch~$t$, $M_k(t)$
have been executed and $|Y(t)|_k$ wait in the backlog. As
$Q(t) = N(t) - \sum_{u<t} A\,s(u)$ by \cref{eq:vqml}, this gives
$\sum_k |Y(t)|_k A_k = \bigl(N(t) - Q(t)\bigr) - A\,M(t)$, which is $X(t) - Q(t)$
by~\labelcref{eq:balance}.

(b) The backlog is empty at $t = 0$. At the beginning of a later epoch, the items present
cover no planned matching of the backlog: the scan of the previous epoch found each of
them uncovered, and the executions that followed only removed items. Each planned matching of the backlog is thus \emph{blocked}
by some class~$i$, of which it needs more than the $X_i(t)$ items present; choose one such
class for each of them. If $r \geq 1$ planned matchings are blocked by class~$i$, they
need together at least $r\,(X_i(t) + 1) \geq X_i(t) + r$ items of that class, so that
$Q_i^-(t) \geq r$ by~(a). Summing over the classes gives~(b).

(c) Summing the coordinates in~(a), and as a planned matching needs at most $a_{\max}$
items,
\[
\|X(t)\|_1 \;=\; \sum_i Q_i(t) + \sum_k |Y(t)|_k\, \|A_k\|_1
\;\leq\; \|Q^+(t)\|_1 - \|Q^-(t)\|_1 + a_{\max}\, |Y(t)| ,
\]
and (b) gives the first bound; the second one holds as $a_{\max} \geq 1$.

(d) If $Q(t) = 0$, the backlog is empty by~(b), and then $X(t) = 0$ by~(a). The converse
is~(a).
\end{proof}

\subsection{The virtual chain drifts inward}\label{sec:virtual}
The drift of the virtual queue is governed by two constants. Under~(iii), fix a positive
solution $\bar\mu$ of $A\bar\mu = \bar\lambda$, where $\bar\lambda = \lambda/\Lambda$, and
let
\[
\delta \;=\; \min_{k \in \Ed} \bar\mu_k , \qquad
c_A \;=\; \min_{\|q\|_1 = 1} \|A^\top q\|_1 , \qquad
\varepsilon \;=\; \delta\, c_A .
\]
The first constant measures how far $\lambda$ sits inside the cone; the best choice is
$\delta = t^\ast/\Lambda$, for the linear program of \cref{sec:stability}. The second one
depends on $A$ alone and is positive because $A$ is surjective, that is, $A^\top$ is
injective; if $s_{\min}$ is the smallest singular value of~$A$, then
$\|A^\top q\|_1 \geq \|A^\top q\|_2 \geq s_{\min} \|q\|_2 \geq s_{\min} \|q\|_1/\sqrt n$, so
that $c_A \geq s_{\min}/\sqrt n$. Finally, summing the rows of $A\bar\mu = \bar\lambda$ gives
$\sum_k \bar\mu_k \|A_k\|_1 = \|\bar\lambda\|_1 = 1$, and $\|A_k\|_1 \geq 1$ for every~$k$,
so that
\begin{equation}\label{eq:mu-budget}
m\,\delta \;\leq\; \|\bar\mu\|_1 \;\leq\; 1 .
\end{equation}

\begin{lemma}[Foster--Lyapunov drift for the virtual chain]\label{lem:foster}
Assume (iii), let $\delta$, $c_A$ and $\varepsilon = \delta c_A$ be as above, and set
$R = 1 + 2a_{\max}$. Whatever the choice of maximizers in \cref{eq:vqml}, as long as
it depends only on the arrivals before epoch~$t$ and on randomness independent of the
arrivals, the virtual queue satisfies, for every $t \in \N$,
\begin{equation}\label{eq:drift}
\E\big[\|Q(t+1)\|_2^2 - \|Q(t)\|_2^2 \,\big|\, \mathcal{F}_t\big]
\;\leq\; R^2 - 2\delta\, \|A^\top Q(t)\|_1
\;\leq\; R^2 - 2\varepsilon \|Q(t)\|_1 ,
\end{equation}
where $\mathcal{F}_t$ is generated by the arrivals before epoch $t$ and the decisions up to
epoch $t$ (the decision $s(t)$ is taken before the arrival $a(t)$ lands). Consequently,
from $Q(0) = 0$,
\begin{equation}\label{eq:cesaro-q}
\frac1T \sum_{t<T} \E\|Q(t)\|_1 \;\leq\; \frac{R^2}{2\varepsilon}
\qquad\text{for every } T \geq 1 .
\end{equation}
\end{lemma}

\begin{proof}
Write $q = Q(t)$ and $s = s(t)$, let $g_k = \langle q, A_k\rangle$ be the score of
hyperedge~$k$, and let $g_+ = \max(0, \max_k g_k)$. The maximum in \cref{eq:vqml} is
$2g_+$, reached by two copies of a best-scoring hyperedge if some score is positive, and
by planning nothing otherwise. As $\bar\lambda = A\bar\mu$,
\[
\langle q, A s - \bar\lambda\rangle
\;=\; 2g_+ - \sum_{k} \bar\mu_k\, g_k
\;=\; \sum_{k} \bar\mu_k\, (g_+ - g_k) \;+\; \bigl(2 - \|\bar\mu\|_1\bigr)\, g_+ .
\]
Both terms are nonnegative, as $g_+ \geq g_k$ and $g_+ \geq 0$. Using $\bar\mu_k \geq \delta$
in the first one and $2 - \|\bar\mu\|_1 \geq 1 \geq m\delta$ in the second one, which
is~\labelcref{eq:mu-budget},
\begin{equation}\label{eq:score-gap}
\langle q, A s - \bar\lambda\rangle
\;\geq\; \delta \sum_{k} (2g_+ - g_k)
\;\geq\; \delta \sum_{k} |g_k|
\;=\; \delta\, \|A^\top q\|_1
\;\geq\; \varepsilon\, \|q\|_1 ,
\end{equation}
where $2g_+ - g_k \geq |g_k|$ because $g_+ \geq \max(g_k, 0)$, and the last inequality is
the definition of~$c_A$.

The arrival $a(t)$ has law $\bar\lambda$ and is independent of $\mathcal{F}_t$, while $q$
and $s$ are $\mathcal{F}_t$-measurable. With $\Delta = a(t) - A s$, of norm
$\|\Delta\|_2 \leq \|\Delta\|_1 \leq 1 + 2a_{\max} = R$,
\[
\E\big[\|Q(t+1)\|_2^2 - \|Q(t)\|_2^2 \,\big|\, \mathcal{F}_t\big]
\;=\; 2 \langle q,\ \bar\lambda - A s \rangle + \E\big[\|\Delta\|_2^2 \,\big|\, \mathcal{F}_t\big]
\;\leq\; R^2 - 2\delta\, \|A^\top q\|_1 ,
\]
which is \cref{eq:drift}. For \cref{eq:cesaro-q}, all expectations are finite, as
$\|Q(t)\|_1 \leq tR$; taking expectations in \cref{eq:drift} and summing over $t < T$ gives
$0 \leq \E\|Q(T)\|_2^2 \leq T R^2 - 2\varepsilon \sum_{t<T} \E\|Q(t)\|_1$.
\end{proof}

\begin{remark}[Why two copies]\label{rem:budget}
The budget of two in \cref{eq:vqml} is used once, in the bound
$2 - \|\bar\mu\|_1 \geq 1$ of the proof above. With a budget of one, that term becomes
$(1 - \|\bar\mu\|_1)\, g_+$, and $\|\bar\mu\|_1 = 1$ exactly when every hyperedge is a
mono-edge. As soon as some hyperedge consumes two items or more, $\|\bar\mu\|_1 < 1$ and one
planned copy per epoch would do, with $\delta$ replaced by
$\min\bigl(\delta, (1 - \|\bar\mu\|_1)/m\bigr)$ in~\labelcref{eq:score-gap}. When every
hyperedge is a mono-edge, a unit budget gives no drift at all: each epoch brings one item
and removes at most one unit, so the total mass $\sum_i Q_i(t)$ is nondecreasing; for
$A = I_2$ and $\lambda = (1,1)$, a unit budget that plans nothing at zero scores leaves
$0$ at the first arrival and is trapped in $\{(1,0), (0,1)\}$. Two copies cover every
instance with one rule; they also make every planned move remove at least two units, which
the accessibility argument of Appendix~\ref{app:pr} uses.
Compare~\cite[\S3.6.2]{NS19}, where two matchings per arrival are invoked precisely when
discussing systems that can be stabilized without single-type matchings (mono-edges, in
our terms).
\end{remark}

\subsection{Stability of \texorpdfstring{\vqml{}}{VQML}}\label{sec:slaving}
\begin{proposition}[\vqml{} stabilizes every interior instance]\label{prop:vqml-stable}
Assume (iii), and let $\varepsilon = \delta c_A$ and $R = 1 + 2a_{\max}$ be as in
\cref{lem:foster}. From the empty start, for every $T \geq 1$,
\begin{equation}\label{eq:cesaro}
\frac1T \sum_{t<T} \E\|X(t)\|_1 \;\leq\; \frac{a_{\max} R^2}{2\varepsilon} ,
\end{equation}
so that $\frac1T \sum_{t<T} \PP(\|X(t)\|_1 \geq K) \leq a_{\max} R^2 / (2\varepsilon K)$
for every $K > 0$ and $T \geq 1$: \vqml{} stabilizes $(G, \lambda)$ in the sense
of~\labelcref{eq:stable}.
\end{proposition}

\begin{proof}
Pathwise, $\|X(t)\|_1 \leq a_{\max}\|Q(t)\|_1$ by \cref{lem:pathwise}(c); averaging over
$t < T$ and using \cref{eq:cesaro-q} gives \cref{eq:cesaro}, and Markov's inequality gives
the tail bound, which vanishes as $K \to \infty$ uniformly in~$T$.
\end{proof}

The bound is crude: on the candy family of \cref{sec:candy} it exceeds the simulated
mean queue by several orders of magnitude. Its content is the scaling: the mean queue is
$O(1/\varepsilon)$ in an interior slack that depends only on $A$ and $\lambda$, and the
system degrades no faster than that as $\lambda$ approaches the boundary of the cone.
This order is the usual consequence of a quadratic drift inequality, a mean backlog of
order $1/\varepsilon$ at distance $\varepsilon$ inside the achievable region~\cite{GNT06};
on simple graphs, the same order is known for match-the-longest~\cite[Lemma~1]{CMV25} and
for generalized max-weight policies~\cite{JMRS20}, whereas no explicit bound was available
for EGPD.

\subsection{Proof of \texorpdfstring{\Cref{thm:main}}{the main theorem}}\label{sec:putting-things-together}

\begin{proof}[Proof of \cref{thm:main}]
Conditions (ii), (iii) and (iii$'$) are equivalent by \cref{lem:geometry}, of which they
are items (c), (a)--(b) and (b$'$). (i) $\Rightarrow$ (ii) is \cref{prop:necessity}, and
(i$'$) $\Rightarrow$ (i) was noted in \cref{sec:stability_notions}: positive recurrence
implies stability. (iii) $\Rightarrow$ (i) is \cref{prop:vqml-stable}, and
(iii) $\Rightarrow$ (i$'$) is \cref{prop:vqml-pr} of Appendix~\ref{app:pr}, for the same
policy. As \vqml{} is defined from $G$ alone, it makes $(G,\lambda)$ positive recurrent,
hence stable, for every $\lambda$ that satisfies~(iii), that is, for every stabilizable
problem: it is maximally stable.
\end{proof}

\subsection{Comparison with concurrent work}\label{rem:nb26}

A first version of the present paper was posted on July 21, 2026 (HAL: hal-05697834;
arXiv:2607.18935). Nguyen and Bu\v{s}i\'c posted the independent preprint~\cite{NB26}
on July 26, 2026. Both works extend~\cite[Prop.~3]{CMV25} to hypergraphs, in
different directions.

Their work generalizes the incidence matrix and the arrival model, with real-valued
multiplicities, batch arrivals and continuous time~\cite{NB26b}, and general
reneging~\cite{NB26c}. It also introduces an intra-class
criterion~\cite[Theorem~3.1]{NB26} that has no direct counterpart here. The present paper
stays with integer multiplicities and studies them in depth: matching rates of stable
policies, an explicit mean-queue bound, the failure of greedy policies, exact regions of
classical families, and the structure of the kernel. The two works also differ in their
classes of policies, their notions of stability and their maximally stable policies; we
compare them below.

Their account of the present paper~\cite[v2, \S1.4.1]{NB26} contrasts their Condition~(3)
with our support form~(iii$'$), numbered (iv) in the version they discuss; their
Condition~(3) is in fact condition~(iii) of \cref{thm:main}, stated since our first version.

The concurrent work~\cite{NB26} characterizes stabilizability with a notion of stability
based on clearing times. Its policies use the history $\mathcal H_T$ up to step~$T$, made
of the arrivals up to~$T$ and of the decisions before~$T$, possibly with randomization,
and its definition allows an arbitrary initial buffer. A policy stabilizes the system if,
from every step~$T$, the expected time $\tau_T$ until all queues are empty is finite
given~$\mathcal H_T$; in our time scale, the state after their step~$T$ is $X(T+1)$.
Their notion and ours, \labelcref{eq:stable}, are not nested. With a single class and its
mono-edge, the policy that discards every item present, the arriving one included, at
the epochs $10^k$, $k \in \N$, has finite clearing times from every step, yet large queues
most of the time; the policy that always keeps exactly one item keeps the queues bounded,
yet never empties the system. Both examples are time-dependent or reducible: for a
Markovian policy without internal state whose chain is irreducible, both notions amount to
positive recurrence.

The two notions define, however, the same stabilizable problems, from the empty state,
and \vqml{} is maximally stable for both. If~(iii) holds, \vqml{} with the canonical rule
of Appendix~\ref{app:pr} makes the chain positive recurrent on the class of~$\emptystate$,
which contains every state reached from~$\emptystate$ (\cref{prop:vqml-pr}); the state
after step~$T$ being a function of the arrivals up to~$T$, $\E[\tau_T \mid \mathcal H_T]$
is at most the expected hitting time of~$\emptystate$ from it, which is finite.
Conversely, suppose that (ii) fails through~$y$, from any initial state, and let
$Z_t = \langle y, N(t)\rangle$. By~\labelcref{eq:balance},
$\langle y, X(t)\rangle - Z_t$ is nondecreasing, so that empty queues at an epoch
$t \geq s$ require $Z_t - Z_s \leq -\langle y, X(s)\rangle \leq -c_s$, with
$c_s = \langle y, X(0)\rangle + Z_s$. If $\langle y, \lambda\rangle > 0$, then
$Z_t \to \infty$ almost surely, the queues are empty only finitely often, and
$\tau_T = \infty$ with positive probability for $T$ large, so that
$\E[\tau_T \mid \mathcal H_T] = \infty$ with positive probability, whatever the policy. If
$\langle y, \lambda\rangle = 0$, some coordinate of~$y$ is positive, so that $c_s > 0$
with positive probability for $s$ large. The level $c_s$ is determined by
$\mathcal H_{s-1}$, and, the policy being non-anticipative, the walk
$(Z_{s+u} - Z_s)_{u \geq 0}$ is centered and independent of~$\mathcal H_{s-1}$. The queues
cannot be empty before this walk first reaches $-c_s$ or below, a time of infinite mean on
$\{c_s > 0\}$ by Wald's identity, so that $\E[\tau_{s-1} \mid \mathcal H_{s-1}] = \infty$
with positive probability. In both cases, no policy stabilizes $(G, \lambda)$ in the sense
of~\cite{NB26}, and our \cref{thm:main} holds with their notion. Their notion conditions on
the history and presupposes non-anticipative policies; ours applies to every policy, and
so does the necessity half of \cref{thm:main}, with the bounds of
\cref{prop:necessity}.

The sufficiency parts rely on different policies. Both commit capacity before hyperedges
can be activated physically, but from opposite sides: the online assignment policies
of~\cite{NB26} earmark each arriving item for a hyperedge type, whereas
\vqml{}~\cite{CMV25}, the rewardless counterpart of EGPD~\cite{NS19}, selects hyperedges
ahead of the items and earmarks no item: a selected hyperedge is activated once the items
present can complete it (\cref{sec:fullstate}).

\section{Conclusion and open questions}\label{sec:conclusion}

Whether a hypergraph matching system can be stabilized is decided, as on simple graphs, by
its incidence matrix and its arrival rates alone: it can exactly when $\lambda$ lies in the
interior of the cone spanned by the hyperedges, which one linear program and a rank
computation decide (\cref{thm:main}). The statement carries over from simple graphs, and so
does the necessity argument, the pathwise bound of~\cite{MM16}, which holds for every
policy; outside the region, the queues are in fact small on a vanishing fraction of
epochs, under every policy (\cref{prop:necessity}). The route to sufficiency does not: on simple graphs, a greedy policy,
match-the-longest, is maximally stable, whereas on hypergraphs no greedy policy is
(\cref{prop:candy-open}); maximal stability then rests on \vqml{}~\cite{CMV25}, the rewardless counterpart of EGPD~\cite{NS19}, an idling policy that needs no knowledge of~$\lambda$. Nor does the
combinatorial structure of the kernel of~$A$: the even cycles and kayak paddles of simple
graphs give way to exchanges whose weights can grow exponentially with~$n$, as on the
pagodas (\cref{prop:pagoda}). The matching rates of a stable policy, whatever the policy,
solve the conservation equation on a surjective support (\cref{prop:rates}), and approaching
a solution whose support is not surjective costs mean queues of order at least $1/\epsilon$
(\cref{cor:degenerate}). For system design, the criterion separates what the compatibility
structure allows from what a given policy achieves, and it yields the exact regions of
classical families, such as the $k$-partite quantum switch and the Fano plane
(\cref{sec:families}).

Several questions remain open. On the complete $k$-uniform hypergraph, that is, for the
$k$-partite quantum switch, all greedy policies coincide (\cref{rem:complete-greedy}), and
whether this single chain is stable on the whole region is known for $k \leq
3$~\cite{MM16,VGNT21,RM21} and, for every $k$, at uniform rates~\cite{NVGT20}; unequal rates
with $k \geq 4$ remain open. Other
families raised in~\cite{Rahme21} still lack an exact region, for instance the complete
$3$-uniform hypergraphs deprived of disjoint hyperedges, proposed there as models of
kidney exchange with three-way cycles and known only through an inner bound for
match-the-longest~\cite[Thm.~4.2]{Rahme21}.

Spontaneous departures, such as the decoherence of stored qubits or patients leaving an
exchange program, are outside the model. Unlike discards (\cref{rem:discard}), which a
policy chooses and which can only help, they are imposed, and they need not preserve
stabilizability. They can be modeled by appending an environment-driven departure step
after the policy's move, but they break the pathwise slaving on which \vqml{} rests
(\cref{lem:pathwise}), so a maximally stable policy would have to be re-derived. On the
candy, for instance, strong departures at the bridge classes $3$ and $5$ should starve the
hyperedge $\{3,4,5\}$, the only outlet of class~$4$, and make $X_4$ diverge under every
policy, greedy or idling; we do not prove this. For hypergraph matching with general
reneging, Nguyen and Bu\v{s}i\'c~\cite{NB26c} show that stability then depends on the
reneging mechanism itself, not only on the rates.

Finally, the mean-queue bound of
\cref{prop:vqml-stable} is crude, and sharper delay guarantees over the whole stability
region, including degenerate instances outside general position (\cref{rem:wall}), are
left for future work.

\newpage

\appendix

\section{Positive recurrence of \texorpdfstring{\vqml{}}{VQML}}\label{app:pr}
This appendix proves (iii)\,$\Rightarrow$\,(i$'$) of \cref{thm:main}. Throughout, ties in
\cref{eq:vqml} are broken by the \emph{canonical rule}: plan two copies of the
best-scoring hyperedge, of smallest index among ties, when some score is positive, and
nothing otherwise. The virtual queue is then a time-homogeneous Markov chain on~$\Z^n$,
the \emph{virtual chain}. The drift of \cref{lem:foster} controls its returns to a finite
set; positive recurrence of the empty state asks in addition that the origin be
accessible, which is where the rule enters (\cref{lem:access,rem:tiebreak}).

\begin{lemma}[The origin is accessible from every state]\label{lem:access}
Assume that no row of $A$ is zero, as is the case under~(iii). Under the canonical rule,
or under any rule that selects a maximizer of \cref{eq:vqml} as a function of the virtual
queue alone and plans nothing when no score is positive (the \emph{idle clause}), the
origin is accessible from every state $q \in \Z^n$ of the virtual chain.
\end{lemma}

\begin{proof}
Call a state $q$ \emph{active} if some score $\langle q, A_k\rangle$ is positive, and
\emph{idle} otherwise. Two facts. A state $q \geq 0$, $q \neq 0$, is active: some
$q_i \geq 1$, row~$i$ has an entry $A_{i,k} \geq 1$, and
$\langle q, A_k\rangle \geq q_i A_{i,k} \geq 1$. At an active state, every maximizer~$s$
of \cref{eq:vqml} has $\|s\|_1 = 2$, as doubling a best column scores twice the best
score, which no $s$ with $\|s\|_1 \leq 1$ can do; hence $\|As\|_1 \geq 2$, the columns
being nonzero.

Let $q \neq 0$, let $s$ be the move of the rule at~$q$, let $r = q - As$, and let the
arrival be of a class~$j$ with $r_j$ minimal, an event of positive probability; the next
state is $r + e_j$. We claim that the pair $\bigl(\|q^+\|_1, \|q^-\|_1\bigr)$ then
decreases strictly in the lexicographic order of~$\N^2$. As this order has no infinite
decreasing sequence, iterating reaches~$0$ in finitely many steps.

\emph{Idle $q$.} The rule plans nothing, so that $r = q$. This vector has a negative
coordinate, else $q$ would be active, hence $r_j \leq -1$: the arrival leaves $q^+$
unchanged and lowers $\|q^-\|_1$ by one.

\emph{Active $q$.} As $0 < \langle q, As\rangle = \sum_i q_i (As)_i$, some class~$i$ has
$q_i \geq 1$ and $(As)_i \geq 1$, so that $\|r^+\|_1 \leq \|q^+\|_1 - 1$. If
$r_j \leq -1$, the arrival does not change the positive part:
$\|(r + e_j)^+\|_1 = \|r^+\|_1 < \|q^+\|_1$. Otherwise $r \geq 0$, hence $q \geq 0$, and
$\|(r + e_j)^+\|_1 = \|q\|_1 - \|As\|_1 + 1 \leq \|q^+\|_1 - 1$.
\end{proof}

\begin{remark}[Tie-breaking matters, but only through the idle clause]\label{rem:tiebreak}
Positive recurrence on the class of $0$ does require a condition on the selection rule.
For $n = m = 2$, $A = I_2$ (two mono-edges) and $\lambda = (1,1)$, which satisfy (iii), the
deterministic rule that prefers $2e_1$, then $2e_2$, to idling at zero-score ties is a
legal selection of maximizers under which the set reachable from $0$ has nine states, $0$
is visited exactly once, and the chain is absorbed into a seven-state class
(\cref{fig:trap}; exhaustive search over the reachable set in the companion
notebook~\cite{candynb}). The drift bound~\labelcref{eq:drift} is unaffected, so \vqml{}
remains stable in the sense of~\labelcref{eq:stable} under this rule; only the class
structure changes, and with it condition~(i$'$): the absorbing class is positive recurrent,
but the empty state is transient and has probability~$0$ under the stationary law.
\Cref{lem:access} shows that the \emph{idle clause} is the only feature of the rule that
matters: the adversarial rule above plans zero-score hyperedges at ties, violating
precisely that clause. Nazari and Stolyar make idling available by convention (the empty
matching $\langle\emptyset\rangle$ is an element of their matching set~\cite[\S3.1]{NS19})
and need no selection rule for their purposes, since stability in their sense holds for
every maximizer.
\end{remark}

\begin{figure}[t]
\centering
\begin{tikzpicture}[x=2.2cm, y=1.7cm, thick,
  st/.style={draw, rounded corners=1.5pt, fill=black!10, inner sep=2.5pt,
             font=\footnotesize},
  tr/.style={st, fill=white, densely dashed},
  act/.style={font=\scriptsize, text=black!60},
  a1/.style={-stealth},
  a2/.style={-stealth, densely dashed}]
\begin{scope}[on background layer]
  \draw[black!20, -stealth] (-2.55, 0) -- (0.5, 0)
    node[right, font=\scriptsize, text=black!50] {$q_1$};
  \draw[black!20, -stealth] (0, -2.5) -- (0, 1.4)
    node[above, font=\scriptsize, text=black!50] {$q_2$};
\end{scope}
\node[tr, label={[act]above right:$2e_1$}] (p00)  at (0,0)   {$(0,0)$};
\node[st, label={[act]above:$2e_2$}]       (m10)  at (-1,0)  {$(-1,0)$};
\node[st, label={[act]left:$2e_2$}]        (m20)  at (-2,0)  {$(-2,0)$};
\node[tr, label={[act]above:$2e_2$}]       (m21)  at (-2,1)  {$(-2,1)$};
\node[st, label={[act]right:$2e_1$}]       (p0m1) at (0,-1)  {$(0,-1)$};
\node[st]                                  (m1m1) at (-1,-1) {$(-1,-1)$};
\node[st]                                  (m2m1) at (-2,-1) {$(-2,-1)$};
\node[st, label={[act]right:$2e_1$}]       (p0m2) at (0,-2)  {$(0,-2)$};
\node[st]                                  (m1m2) at (-1,-2) {$(-1,-2)$};
\draw[a1] (p00) -- (m10);
\draw[a2] (p00) -- (m21);
\draw[a1] (m10) -- (p0m2);
\draw[a2] (m10) to[bend left=15] (m1m1);
\draw[a1] (m21) -- (m1m1);
\draw[a2] (m21) -- (m20);
\draw[a1] (p0m2) to[bend left=15] (m1m2);
\draw[a2] (p0m2) -- (m2m1);
\draw[a1] (m1m1) to[bend left=15] (p0m1);
\draw[a2] (m1m1) to[bend left=15] (m10);
\draw[a1] (m20) -- (m1m2);
\draw[a2] (m20) to[bend left=15] (m2m1);
\draw[a1] (m1m2) to[bend left=15] (p0m2);
\draw[a2] (m1m2) -- (m1m1);
\draw[a1] (m2m1) -- (m1m1);
\draw[a2] (m2m1) to[bend left=15] (m20);
\draw[a1] (p0m1) to[bend left=15] (m1m1);
\draw[a2] (p0m1) -- (m20);
\end{tikzpicture}
\caption{The adversarial chain of \cref{rem:tiebreak} ($A = I_2$,
$\lambda = (1, 1)$), each state drawn at its coordinates
$(q_1, q_2)$. Solid arrows: class-$1$ arrival; dashed arrows: class-$2$
arrival (probability $\tfrac12$ each). Labeled states carry the multiset that the
rule plans there; at unlabeled states it plans nothing. Every labeled move has zero
score, except at $(-2,1)$ where $2e_2$ is the unique maximizer: whenever
the chain touches an axis, the rule plans a zero-score doubled
mono-edge, which the idle clause forbids, and overshoots into the
nonpositive orthant (except for $(0,0) \to (-2,1)$, resorbed at the next
step). The origin (visited once) and $(-2,1)$, hollow, are
transient; the seven filled states form the absorbing class. The drift
bound~\labelcref{eq:drift} is unaffected, so \vqml{} remains stable in the sense
of~\labelcref{eq:stable}; the empty state, transient, has stationary probability~$0$,
and condition~(i$'$) fails for it.}
\label{fig:trap}
\end{figure}

\begin{proposition}[\vqml{} is positive recurrent]\label{prop:vqml-pr}
Run \vqml{} with the canonical rule, or with any rule as in \cref{lem:access}.
\begin{enumerate}
\item[(a)] Under~(iii), \vqml{} is a Markovian policy that makes $(G, \lambda)$ positive
recurrent: the chain $(S(t))$ is positive recurrent on the communicating class
of~$\emptystate$, which is the set of states reachable from~$\emptystate$.
\item[(b)] Under~(iii), with $\varepsilon$ and $R$ as in \cref{lem:foster}, the stationary
law~$\pi$ of $(S(t))$ on that class satisfies
\[
\E_{\pi}\|Q\|_1 \;\leq\; \frac{R^2}{2\varepsilon}
\qquad\text{and}\qquad
\E_{\pi}\|X\|_1 \;\leq\; \frac{a_{\max} R^2}{2\varepsilon},
\]
where $Q = X - \sum_k |Y|_k A_k$ is the virtual queue of the state $(X, Y)$. By Little's
law, the mean sojourn time of an item is then at most $a_{\max}R^2/(2\varepsilon\Lambda)$.
\end{enumerate}
\end{proposition}

\begin{proof}
\emph{Markov structure.} A backlog is a finite word over the finite alphabet~$\Ed$, so the
state space of $(S(t))$ is countable. The decision $s(t)$ is a function of $Q(t)$, hence
of $S(t)$ by \cref{lem:pathwise}(a), and $S(t+1)$ is a function of $S(t)$, $s(t)$ and
$a(t)$ (\cref{sec:fullstate}). The arrivals being i.i.d., $(S(t))$ is a time-homogeneous
Markov chain, whose internal state~$Y$ is updated at arrival epochs only. The physical
activations are the executions of planned matchings; they occur at arrival epochs and
remove items that are present. Hence \vqml{} is a Markovian policy in the sense of
\cref{sec:stability_notions}.

\emph{Virtual chain.} Let $\mathcal{R}$ be the set of the states of the virtual chain
that are reachable from~$0$. It is closed, and each of its states leads back to~$0$ by
\cref{lem:access}: the chain restricted to~$\mathcal{R}$ is irreducible. By
\cref{eq:drift}, whose left-hand side now depends on $Q(t)$ alone, the drift of
$\|\cdot\|_2^2$ is at most $-1$ outside the finite set
$\{q \in \mathcal{R} : \|q\|_1 \leq (R^2+1)/(2\varepsilon)\}$, on which it is at most~$R^2$.
Foster's criterion~\cite{B99} gives positive recurrence: from $Q(0) = 0$, the return time
$\tau_0 = \inf\{t \geq 1 : Q(t) = 0\}$ has a finite mean.

\emph{Physical chain.} By \cref{lem:pathwise}(d), $S(t) = \emptystate$ if and only if
$Q(t) = 0$, along every trajectory from the empty start. The return time of $(S(t))$
to~$\emptystate$ is thus~$\tau_0$, of finite mean: $\emptystate$ is a positive recurrent
state. Its communicating class is then closed, so that it is the set of states reachable
from~$\emptystate$, and the chain is positive recurrent on it. So is the continuous-time
chain, its uniformization at rate~$\Lambda$, with the same stationary law. This is~(a).

\emph{Stationary bounds.} As recalled in \cref{sec:stability_notions}, the Ces\`aro
averages of the laws of $S(t)$, from the empty start, converge pointwise to~$\pi$. By
Fatou's lemma and \labelcref{eq:cesaro},
\[
\E_{\pi}\|X\|_1 \;\leq\; \liminf_{T \to \infty} \frac1T \sum_{t<T} \E\|X(t)\|_1
\;\leq\; \frac{a_{\max} R^2}{2\varepsilon} ,
\]
and likewise for $\|Q\|_1$ with \labelcref{eq:cesaro-q}. For Little's law, the
continuous-time chain has stationary law~$\pi$, items arrive at rate~$\Lambda$, and every
item present leaves before the next return to~$\emptystate$: the mean number of items in
the system is $\Lambda$ times the mean sojourn time.
\end{proof}

By~(a), \cref{prop:rates} applies to \vqml{}: its mean matching rates converge, to a
solution of the conservation equation with a surjective support.

\end{document}